\newcommand{\notes}{1}

\documentclass[11pt]{article}

\usepackage{sty}
\usepackage{tikz}
\usetikzlibrary{positioning,fit,arrows.meta}

\title{Testing $k$-submodularity}
\author{
Themistoklis Haris \\ Boston University \and
Diptaksho Palit\thanks{D.P.\ was supported by the U.S. National Science Foundation under Grant No.\ 2022446.} \\ Boston University
}
\date{\today}

\begin{document}

\maketitle

\begin{abstract}
We initiate the study of property testing for $k$--submodular functions, a higher-dimensional analogue of submodular functions defined on partial partitions of a ground set.
While $k$--submodularity retains the diminishing-returns flavor of ordinary submodularity, it also introduces a pairwise monotonicity constraint comparing competing assignments of the same element.
This additional local structure makes the testing problem qualitatively different from the classical case.

Our results show a sharp contrast between distance regimes.
In the $\ell_p$ regime for $p \geq 1$, we prove that every bounded $k$--submodular function is close to a junta on the hypergrid.
Combined with an implicit-learning tester for hypergrid domains, this yields a constant-query tester for $k$--submodularity.
In the Hamming distance regime, $k$--submodularity admits two qualitatively different local witnesses---violated \emph{squares} for diminishing marginal gains, and violated \emph{triangles} for pairwise-monotonicity failures---and the latter has no counterpart at $k=1$.
We prove density theorems for both witness types via repair on filters and ideals of partial partitions, yielding non-adaptive, one-sided sub-exponential-query testers for the two component properties of $k$--submodularity.
We then exhibit a configuration in which the two repair directions are forced into opposition on a shared vertex, identifying a structural barrier to combining these into a tester for the full property.

Finally, for bounded-range functions, we give an adaptive tester for monotone $k$--submodularity via a pseudo-DNF representation and learning on the hypergrid.
Several of the structural and learning tools developed here may be useful for testing other properties over product domains.
\end{abstract}

\thispagestyle{empty}
\setcounter{page}{0}

\newpage
\tableofcontents

\thispagestyle{empty}
\setcounter{page}{0}

\newpage

\section{Introduction}
Submodularity is one of the canonical examples of a global property with a simple local interpretation.
A set function is submodular when it satisfies the principle of \emph{diminishing marginal returns}: the incremental value of adding an element can only decrease as the set to which it is added grows.
Formally, a set function $f : 2^{[n]} \to \mathbb{R}$ is said to be \emph{submodular} if, for all subsets $X, Y \subseteq [n]$, it satisfies the inequality:
\begin{align*}
f(X) + f(Y) \geq f(X \cap Y) + f(X \cup Y).
\end{align*}

This local-to-global nature makes submodularity a natural target for property testing \cite{goldreich1998property}, where one is given oracle access to an unknown function and must distinguish functions satisfying a property from those that are far from satisfying it.
Yet even for ordinary submodularity, the query complexity of testing remains poorly understood.
Seshadhri and Vondrák \cite{seshadhri2014submodularity} gave a tester using roughly $(1/\varepsilon)^{\sqrt{n}}$ queries and proved a linear lower bound; the latter was later improved to a nearly quadratic lower bound by Hatami and Vondrák \cite{hatamitesting}.
The resulting exponential gap has remained open, and has motivated the study of more structured regimes, such as bounded-range submodular functions \cite{RaskhodnikovaY13} and testing under $\ell_p$ distance rather than Hamming distance \cite{blais2016testing}.

In this work we ask what happens to this picture for a higher-dimensional analogue of submodularity.
For a fixed integer $k \geq 1$, a point in $[k+1]^n$ can be viewed as a \emph{partial partition} of the ground set $[n]$: each element is either assigned to one of $k$ parts, or left unassigned.
A function on these partial partitions is \emph{$k$--submodular} if marginal gains diminish no matter which part an element is assigned to.
The case $k=1$ is ordinary submodularity, while already for $k=2$ the domain has multiple incompatible directions in which a single element may be added.
This function class has been studied extensively in optimization \cite{ohsaka2015monotone, ene2022streaming}, with applications including tensor placement \cite{ohsaka2015monotone} and online ad allocation \cite{feldman2010online}.

From a testing perspective, $k$--submodularity is not merely submodularity on a larger domain.
A useful characterization, due to Ward and Živný \cite{ward2016maximizing}, decomposes $k$--submodularity into two conditions.
The first is \emph{diminishing marginal gains}, which generalizes the familiar submodular inequality inside each direction of the hypergrid.
The second is \emph{pairwise monotonicity}, a constraint comparing the marginal gains of assigning the same element to two different parts.
One way to understand this decomposition is to separate two kinds of local interactions.
When two partial partitions make compatible choices, $k$--submodularity behaves like ordinary diminishing returns along the corresponding direction.
When they make incompatible choices for the same element, the join operation discards that element from both parts; the resulting inequality says that assigning the element to two different parts cannot have two strongly negative marginal gains at the same time.
This second condition has no analogue when $k=1$, and it is the source of much of the new combinatorial structure in the testing problem.

We initiate the study of testing $k$--submodularity.
Our results show a sharp contrast between two distance regimes.
In the $\ell_p$ regime, analytic structure survives the passage from the Boolean hypercube to the hypergrid: we prove a junta-approximation theorem for $k$--submodular functions and combine it with an implicit-learning tester on hypergrids.
In the Hamming distance regime, the local structure fragments: we obtain testers for the two local constraints by developing square and triangle witness structures, and show that their repair procedures expose a genuine obstruction to the standard local-repair paradigm for full $k$--submodularity.
Along the way, we develop testing and learning tools for hypergrid domains that may be useful beyond this problem.

\subsection{Our Results}

We organize our contributions around three questions.
First, does the analytic structure that makes submodularity testable in $\ell_p$ distance survive on the richer domain of partial partitions?
Second, in the Hamming distance regime, can the local-violation framework for submodularity be adapted to the two distinct constraints that characterize $k$--submodularity?
Third, can bounded range assumptions recover efficient testability for a meaningful subclass of $k$--submodular functions?

\paragraph{$\ell_p$ testing from hypergrid junta structure.}
Our first result is a constant-query tester for $k$--submodularity in $\ell_p$ distance.
Here and throughout this section, constant means independent of the ambient dimension $n$.

\begin{theorem}
    For all $p \geq 1$, there exists an adaptive, one-sided error $\ell_p$-tester for $k$-submodularity making  $q = \left(\frac{p}{\varepsilon}\right)^{2} \cdot 2^{\left(\frac{p}{\varepsilon}\right)^{2} \log \left(\frac{p}{\varepsilon}\right)}$ queries.
\end{theorem}


The main structural input is not a black-box reduction to the submodular case.
We prove that every bounded $k$--submodular function is close to a function depending on only a small number of coordinates, even though each coordinate can move in $k$ competing directions.
This requires choosing influential coordinates in a way that respects all parts of a partial partition, and uses the fact that $k$--submodular functions satisfy a self-bounding variance inequality.

\begin{theorem}
    \label{thm:junta-thm}
    Let $f :[k+1]^J\to [0,1]$ be a $k$--submodular function and $\epsilon > 0$ be a small enough constant.
    Then there exists a $k$--submodular function $h : [k+1]^J\to [0,1]$ that depends only on $J'\subseteq J$ dimensions, where 
    $$
        |J'|=O\left(\frac{1}{\varepsilon^2}\cdot \log\frac{1}{\varepsilon}\right),
    $$
    and the function $h$ satisfies
    $$
        \lVert f-h \rVert_2 \leq \varepsilon.
    $$
\end{theorem}

We then develop the testing component needed to turn this structure into an algorithm.
The implicit-learning framework of Blais and Bommireddi \cite{blais2016testing} relies on identifying influential coordinates and checking consistency with a small junta.
We formulate and analyze this framework on hypergrid domains using Fourier analysis over generalized product spaces, and simplify parts of the original argument by separating the learning and final testing samples.

\begin{theorem}
    \label{thm:implicit-learning-testing}
    Let $p \geq 1$. Suppose $\mathcal{P} \subseteq \{f:[k+1]^n \to [0,1]\}$ is a property of functions such that for any $f \in \mathcal{P}$ there exists a $\kappa$-junta $h :[k+1]^n \to [0,1]$ satisfying $||f-h||_2 \leq \varepsilon$. Then, there exists a tester that, given query access to an unknown function $f$, will make $\varepsilon^{-9}\cdot 2^{O(\kappa\log \kappa)}$ queries to $f$ and with probability at least $2/3$ accept if $f \in \mathcal{P}$ and reject if $f$ is $\varepsilon$-far from $\mathcal{P}$ in the $\ell_2$ norm.
\end{theorem}

\paragraph{Local testing in Hamming distance.}
Our second set of results addresses the Hamming-distance regime, where the Seshadhri--Vondr\'ak local-repair paradigm for ordinary submodularity does not extend directly.
Pairwise monotonicity is vacuous when $k=1$, but on the hypergrid its violations form a new witness shape---\emph{triangles}---incomparable with the \emph{squares} that witness diminishing-marginal-gain failures.
Repair likewise shifts from Boolean subcubes to filters and ideals of partial partitions, and the density arguments must be re-proved in this setting.
We obtain non-adaptive, one-sided testers for the two component properties of $k$--submodularity:

\begin{theorem}[\cref{thm:dim-marg-gain-test-ana} and \cref{thm:pair-mon-test-ana}]
    Let $f:[k+1]^n \to \mathbb{R}$ be given via oracle access. There exist non-adaptive, one-sided error testers for the diminishing marginal gains and pairwise monotonicity properties with the following guarantees:
    \begin{itemize}
        \item If $f$ satisfies the property, the tester accepts with probability $1$.
        \item If $f$ is $\varepsilon$-far from the property (i.e., at least $\varepsilon \cdot [k+1]^n$ values must be modified to satisfy it), the tester rejects with probability at least $2/3$.
    \end{itemize}
    The tester for diminishing marginal gains makes
    $q = 4n^2 \cdot (1/\varepsilon)^{\Theta(\sqrt{n \log n})}$ queries;
    the tester for pairwise monotonicity makes
    $q = 3kn \cdot (1/\varepsilon)^{\Theta(\sqrt{n \log n})}$ queries.
\end{theorem}

Notably, only the pairwise-monotonicity tester has query complexity depending on $k$, despite $k$ governing the structure of the domain.

Combined, these results expose a structural difference between $k$--submodularity and ordinary submodularity that the analytic ($\ell_p$) regime does not see.
In \cref{sec:combining-testers} we exhibit an explicit configuration (\Cref{fig:repair-deadlock}) on a high item-weight filter in which repairing a triangle violation forces value changes in the \emph{opposite} direction from those needed to repair an overlapping square violation, on a shared vertex.
The two repair directions are not merely incompatible at one point---they are jointly inconsistent across an entire region.
We record this as \cref{obs:repair-deadlock}: it rules out the natural component-wise extension of local repair to full $k$--submodularity, and isolates a concrete barrier that any sub-exponential local-witness tester must overcome by enlarging the witness vocabulary beyond squares and triangles.



\paragraph{Bounded range and monotone $k$--submodularity.}
Finally, we show that efficient Hamming-distance testing can be recovered for monotone $k$--submodular functions with bounded range.
The result is again obtained through learning, but the relevant representation is no longer a junta.
Instead, we utilize pseudo-DNF representations for functions on the hypergrid and prove that bounded-range monotone $k$--submodular functions have bounded-width pseudo-DNFs.

\begin{theorem}
    \label{thm:mon-k-submod-bnd-rng-lrn}
    Let $\eps, \delta \in (0, 1)$ be constants.
    There exists a tester $\cT$ which, given query access to a function $f : [k+1]^n \to \{ 0, 1, \cdots, r \}$ satisfies the following.
    \begin{enumerate}
        \item If $f$ is monotone $k$--submodular, then $\cT$ {\bf accepts} with probability $\geq 1-\delta$.
        \item If $f$ is $\eps$-far from monotone $k$--submodularity, then $\cT$ {\bf rejects} with probability $\geq 1-\delta$.
        \item $\cT$ makes at most $\textsf{poly} (n, 1/\epsilon, 1/\delta, (kr)^{O(r k^2 \log (r/\eps))})$ queries.
    \end{enumerate} 
\end{theorem}
The learning result for bounded-width pseudo-DNFs uses a hypergrid version of the switching-lemma/Fourier-sparsity route and yields a tester via the Kushilevitz-Mansour algorithm.
We expect this learner to be useful beyond $k$--submodularity.

\subsection{Prior Work}

Testing whether a given function is submodular was first examined by Parnas, Ron and Rubinfeld \cite{parnas2003testing}, who studied the limited setting of testing for a class of matrices called \textit{Monge matrices}.
In its full generality, the problem was considered by Seshadhri and Vondrák \cite{seshadhri2014submodularity}.
They developed the local-witness and repair framework for Hamming-distance testing on the Boolean hypercube, proving an upper bound of $(1/\eps)^{O(\sqrt{n \log n})}$ and a linear lower bound.
The latter was later improved to $\Omega(n^2/\log n)$ by Hatami and Vondrák \cite{hatamitesting}.
Blais and Bommireddi \cite{blais2016testing} subsequently studied the problem in the $\ell_p$ norm, and using an approximation result from \cite{feldman2016optimal}, designed an implicit-learning tester with constant query complexity in the input length.
Raskhodnikova and Yaroslavtsev \cite{RaskhodnikovaY13} studied bounded-range submodular functions through a learning-based approach and proved a polynomial-time upper bound in this regime.
Chakrabarty and Huang \cite{chakrabarty2015recognizing} studied the problem of testing for coverage functions, which are a special type of submodular function.

Our work is informed most directly by the three submodularity-testing frameworks of Seshadhri and Vondrák \cite{seshadhri2014submodularity}, Blais and Bommireddi \cite{blais2016testing}, and Raskhodnikova and Yaroslavtsev \cite{RaskhodnikovaY13}.
These works provide the conceptual starting points for the three regimes we study.
The passage to $k$--submodularity introduces new difficulties in each case: local witnesses split into squares and triangles in the Hamming regime, influential coordinates must account for competing assignments in the $\ell_p$ regime, and bounded-range learning requires representations over the hypergrid rather than the Boolean cube.

To our knowledge, there is no prior work on testing $k$-submodularity, though the property has been studied extensively from the optimization perspective \cite{ohsaka2015monotone, ene2022streaming}.
We also use the structural characterization of Ward and Živný \cite{ward2016maximizing}, which identifies $k$--submodularity with the conjunction of diminishing marginal gains and pairwise monotonicity.

\subsection{Technical Overview}

The main technical theme is that the hypergrid introduces several inequivalent notions of locality.
In the $\ell_p$ regime, locality is analytic and is captured by influential coordinates.
In the Hamming regime, locality is combinatorial and is captured by small violated configurations.
For bounded-range functions, locality appears through compact symbolic representations.
In each part of the paper we develop the corresponding notion of locality for partial partitions.

\paragraph{The analytic regime: finding the right influential coordinates.}
The $\ell_p$ tester has two ingredients: a junta-approximation theorem for $k$--submodular functions, and a tester for properties that are close to juntas on the hypergrid.
The obstacle in the first ingredient is that a coordinate has $k$ possible nonzero directions.
Several natural ways of declaring a coordinate influential in one part at a time almost work, but fail to interact correctly with the boosting lemma for down-monotone families (\cref{lemma:boosting}).

Our solution is to define the relevant influence threshold by looking across all partitions of the remaining coordinates.
This lets the selection rule remain compatible with the combinatorial boosting argument while still respecting the multi-directional nature of the domain.
The second key input is a variance bound: we prove that $k$--submodular functions are $2$-self bounding (\cref{lemma:self-bounding-variance-bound}).
Together, these ingredients yield a junta depending on only $O(\varepsilon^{-2}\log(1/\varepsilon))$ coordinates.

Once the junta structure is established, the testing step is conceptually separate.
We analyze implicit learning directly on $[k+1]^n$ using Fourier analysis over product domains.
The tester identifies a candidate set of influential coordinates and then checks whether the observed behavior is consistent with a junta from the target property.
Using an independent sample in the final testing stage also simplifies the decorrelation issues that arise in the original Boolean-domain analysis of \cite{blais2016testing}.

\paragraph{The Hamming regime: local witnesses and repair.}  
In Hamming distance, a tester needs many local witnesses whenever the function is far from the property.
For ordinary submodularity, the relevant witness is a violated square in the Boolean hypercube, and the analysis of Seshadhri and Vondrák \cite{seshadhri2014submodularity} shows that few violated squares imply a global repair.

For $k$--submodularity, a single witness shape is no longer sufficient.
Using the characterization of Ward and Živný \cite{ward2016maximizing}, we separate the property into diminishing marginal gains and pairwise monotonicity.
Diminishing marginal gains is witnessed by generalized violated squares, while pairwise monotonicity is witnessed by violated triangles.
We prove density statements for both witnesses by designing repair procedures on large inclusion-closed regions of the hypergrid: ideals $\{Y \preceq X\}$ and filters $\{Y \succeq X\}$.
This gives sub-exponential testers for the two component properties.

The same repair viewpoint also yields a structural observation about the full Hamming-distance problem (\cref{obs:repair-deadlock}, illustrated in \Cref{fig:repair-deadlock}).
On high item-weight filters, the repair direction needed for a violated square is forced into conflict with the repair direction needed for an overlapping violated triangle on a shared vertex.
This is not a limitation of our proof---it is a structural property of the problem: any tester built on the squares-and-triangles vocabulary, no matter how the repair is organized, must encounter a configuration of this kind.
A sub-exponential tester for full $k$--submodularity in Hamming distance therefore requires an enriched witness vocabulary, not merely a more clever combination of square and triangle repairs.

\paragraph{The bounded-range regime: representation before testing.}
For bounded-range monotone $k$--submodular functions, we take a representation-theoretic route.
We first show that such functions admit bounded-width pseudo-DNF representations over the hypergrid (see \cref{def:pseudo-dnf}).
This representation captures the monotone threshold-like structure imposed jointly by bounded range and diminishing marginal gains.

We then prove a hypergrid switching lemma for pseudo-DNFs, using low-depth $k$-ary decision trees in place of binary trees.
Combined with Fourier analysis on the hypergrid, this implies sparse Fourier approximation for bounded-width pseudo-DNFs.
The resulting learner, together with the Kushilevitz-Mansour algorithm \cite{KushilevitzM93}, yields the bounded-range tester.

\subsection{Discussion and Open Problems}
Our results leave several natural questions open.
The most immediate is whether full $k$--submodularity can be tested in Hamming distance with sub-exponential query complexity.
Our component testers for diminishing marginal gains and pairwise monotonicity show that each constraint has many local witnesses when violated, but the repair procedures for the two constraints can be forced to conflict.
Resolving this---either by exhibiting a richer witness family with a compatible repair procedure, or by proving a sub-exponential lower bound that confirms the barrier is fundamental---is, in our view, the right next step in this regime.

A second direction is to understand the bounded-range setting without monotonicity.
Our tester in this regime relies on a pseudo-DNF representation for monotone $k$--submodular functions, and it is unclear whether a comparable representation exists for non-monotone functions.

Finally, it would be interesting to prove lower bounds for testing $k$--submodularity.
The known lower bounds for submodularity do not transfer directly, since natural embeddings can destroy pairwise monotonicity.
Thus even formulating the right hard instances appears to require using the genuinely multi-part structure of the domain.

\subsection{Organization}

The paper is organized as follows.
In \cref{sec:prelims}, we formally define the problem and related structures on the hypergrid.
In \cref{sec:junta-approx}, we prove the junta approximability theorem for $k$--submodular functions.
In \cref{sec:testing-implicit-learning}, we prove the hypergrid implicit-learning theorem and derive our $\ell_p$ tester.
In \cref{sec:ksubmod-tester}, we develop Hamming-distance testers for \emph{diminishing marginal gains} and \emph{pairwise monotonicity}, and analyze the obstruction to combining their repairs into a local-repair tester for full $k$--submodularity.
In \cref{sec:testing-bounded-range} we prove our $\ell_0$ tester for monotone $k$--submodularity in the bounded range regime.

\section{Preliminaries}
\label{sec:prelims}

\begin{definition}[Partial partition]
    A \emph{partial partition} of $[n]$ into $k$ subsets is defined $X = (X_1,...,X_k)$, where each $X_i \subseteq [n]$ and $X_i \cap X_j = \emptyset$ for all $i\neq j \in [k]$.
    We let $X_{k+1} = [n] \setminus \cup_{i=1}^k X_i$, and this allows us to associate \emph{partial partitions} of $[n]$ into $k$ subsets with the hypergrid $[k+1]^n$ (each element of $[n]$ is associated to a bucket in $[k+1]$, yielding a vector in $[k+1]^n$).
    If $e \notin X_{k+1}$, we say that $e \in X$.
    We impose a \emph{partial order} on \emph{partial partitions} by defining $X \preceq Y$ if and only if $X_i \subseteq Y_i$ for all $i \in [k]$.
    We let $|X|$ denote the \emph{item-weight} of a partial partition, which is the total number of items picked.
    Numerically, this is equal to $\sum_{i = 1}^{k} |X_i|$.
\end{definition}

\begin{definition}[Marginal gains]
    Define the \emph{marginal gains} of a function $f$ on a partial partition $X$ with respect to an \textit{index} $i \in [k]$ and an element $e \in X_{k+1}$ as
    \[
        \Delta_{f}(X,i,e) := f(X_1,...,X_i \cup \{e\},...,X_k) - f(X).
    \]
    Letting $X_{e}^{i} := (X_1,...,X_i \cup \{e\},...,X_k)$, we can rewrite this as $\Delta_f(X,i,e) = f(X_{e}^{i})-f(X)$.     
\end{definition}

\begin{definition}[$k$--submodularity]
    A function $f : [k+1]^n \to \mathbb{R}$ is \emph{$k$--submodular} if, for all partial partitions $X,Y$, the following holds.
    \[
        f(X)+f(Y) \geq f(X \sqcup Y) + f(X \sqcap Y),
    \]
    where $X \sqcap Y := (X_1 \cap Y_1, \cdots, X_k \cap Y_k)$ and $X \sqcup Y = \left((X_1 \cup Y_1) \setminus \bigcup \limits_{i\neq 1} (X_i\cup Y_i), \cdots, (X_k \cup Y_k) \setminus \bigcup \limits_{i\neq k} (X_i\cup Y_i) \right)$.
\end{definition}

The following theorem of Ward and {\v{Z}}ivn{\`y} \cite{ward2016maximizing} gives us an equivalent definition of $k$--submodularity, one that is described by two more intuitive properties.
The intuition is that every local comparison between two partial partitions can be decomposed into compatible and conflicting moves.
Compatible moves are governed by the usual diminishing-returns principle: adding an unassigned element to part $i$ becomes less valuable as the current partial partition grows.
Conflicting moves arise only when the same element could be assigned to two different parts.
For example, comparing $X_e^i$ and $X_e^j$ gives $X_e^i \sqcap X_e^j = X$ and $X_e^i \sqcup X_e^j = X$, so the $k$--submodular inequality becomes
\[
    \Delta_f(X,i,e)+\Delta_f(X,j,e) \geq 0.
\]
Thus pairwise monotonicity is the local condition that prevents two competing assignments of the same element from both being too costly.

\begin{theorem}\cite{ward2016maximizing}
\label{thm:k-submod-alt-def}
A function $f:[k+1]^n\to \mathbb{R}$ is \emph{$k$--submodular} if and only if the following two conditions hold.
\begin{itemize}
    \item \textbf{Diminishing marginal gains.} For all partial partitions $X,Y$ such that $X \preceq Y$, $e \notin Y$, and $i \in [k]$, the following holds.
    \[
        \Delta_f(X,i,e) \geq \Delta_f(Y,i,e).
    \]
    \item \textbf{Pairwise monotonicity.} For all partial partitions $X$, $e \notin X$, and $i \neq j \in [k]$, the following holds.
    \[
        \Delta_f(X,i,e)+\Delta_f(X,j,e) \geq 0.
    \]
\end{itemize}
\end{theorem}

\section{Junta approximation of $k$--submodular functions}
\label{sec:junta-approx}
In the following two sections, we prove the structural and algorithmic ingredients behind our $\ell_p$ tester for $k$--submodularity.
The structural ingredient is a junta approximation theorem for $k$--submodular functions on the hypergrid.
The main difficulty is choosing influential coordinates in a way that accounts for all $k$ directions of a coordinate while still giving a dimension-independent approximation.

This approximation theorem, combined with an implicit learning tester for properties on hypergrids with a good junta approximation, yields the desired tester.

We prove the following approximation theorem.

\begin{theorem}[Reminder of \cref{thm:junta-thm}]
    Let $f :[k+1]^J\to [0,1]$ be a $k$--submodular function and $\epsilon > 0$ be a small enough constant.
    Then there exists a $k$--submodular function $h : [k+1]^J\to [0,1]$ that depends only on $J'\subseteq J$ dimensions, where 
    $$
        |J'|=O\left(\frac{1}{\varepsilon^2}\cdot \log\frac{1}{\varepsilon}\right),
    $$
    and the functions $h$ satisfies
    $$
        \lVert f-h \rVert_2 \leq \varepsilon.
    $$
\end{theorem}

Before we get to the proof of \cref{thm:junta-thm}, we first set up some necessary machinery.

\subsection{Setup}

There are three main parts to the setup.

{\bf Boosting lemma.}
We cite the following boosting lemma due Goemans and Vondrák \cite{goemans2006covering}.

\begin{definition}[Down-monotone family]
    \label{def:down-monotone-set}
    Let $X$ be a finite set and $\mathcal{F} \subseteq 2^X$ be a collection of subsets of $X$.
    The family $\mathcal{F}$ is called down-monotone if the following condition holds.
    \[
        A \in \mathcal{F} \land B \subseteq A \implies B \in \mathcal{F}.
    \]
\end{definition}

\begin{lemma}[\cite{goemans2006covering} Boosting lemma]
    \label{lemma:boosting}
    Let $X$ be a finite set and $\mathcal{F} \subseteq 2^X$ be a down-monotone family.
    For $p \in (0,1)$, define
    \[
        \sigma_p:= \Pr [ X(p) \in \mathcal{F} ],
    \]
    where $X(p)$ is a random subset of $X$ such that each element is picked independently with probability $p$.
    Then
    \[
        \sigma_p = (1-p)^{\phi (p)},
    \]
    where $\phi(p)$ is a non-decreasing function for $p \in (0, 1)$. 
\end{lemma}

{\bf Self--bounding functions and concentration.}
A very useful tool in bounding the variance of submodular functions comes from the literature of concentration for \textit{self--bounding} functions.

\begin{definition}[Self--bounding function]
    \label{def:self-bound}
    A function $f : \mathcal{X}^n \to [0,1]$ is \textbf{$\alpha$-self--bounding} if there exist functions $f_i : \mathcal{X}^{n-1} \to [0,1]$ such that for all $x_1, \cdots, x_n \in \mathcal{X}$ and all $i \in [n]$, the following conditions hold.
    \begin{enumerate}
        \item $0 \leq f(x_1, \cdots, x_n) - f_i(x_1, \cdots, x_{i-1}, x_{i+1}, \cdots, x_n) \leq 1$; and
        \item $\sum\limits_{i=1}^n(f(x_1,...,x_n)-f_i(x_1,x_2,...,x_{i-1},x_{i+1},...,x_n)) \leq \alpha f(x_1,...,x_n)$.
    \end{enumerate}
\end{definition}

\begin{lemma}[\cite{boucheron2003concentration} Bounded variance lemma]
    \label{lemma:self-bounding-variance-bound}
    If a function $f : \mathcal{X}^n \to [0, 1]$ is $\alpha$-self--bounding then
    \[
        \mathrm{Var}[f(X)] \leq \alpha \cdot \mathbb{E} [f(X)].
    \]
\end{lemma}

{\bf Bounded variance of $k$--submodular functions.}
Submodular functions that are Lipschitz are self--bounding, which allows us to bound their variance.
We show that a similar property holds for $k$--submodular functions.

\begin{definition}[Lipschitz function]
    \label{def:lipschitz}
    A function $f : [k+1]^n \to [0,1]$ is $\alpha$-Lipschitz if $|\Delta_f (X, j, i)| \leq \alpha$ for all $i \in [n]$ and $j \in [k]$.
\end{definition}

\begin{lemma}
    \label{lemma:k-submod-bounding}
    Let $f:[k+1]^n\to [0,1]$ be $1$--Lipschitz and $k$--submodular. Then:
    \[
        \mathrm{Var} [f] \leq 2 \cdot \mathbb{E} [f].
    \]
\end{lemma}
\begin{proof}
    We prove this by showing that $f$ is $2$-self--bounding, and invoking \cref{lemma:self-bounding-variance-bound}.
    For $i \in [n]$, define the function $f_i$ as
    $$f_i (x_1, \cdots, x_{i-1}, x_{i+1}, \cdots x_n) := \min_{j \in [k+1]} \{ f(x_1, \cdots, x_{i-1}, j, x_{i+1}, x_n) \}.$$
    Then the first condition
    $$0 \leq f(x_1,...,x_n) - f_i(x_1,...,x_{i-1},x_{i+1},...,x_n) \leq 1$$
    holds trivially (by choice of minimum value and the range of $f$ being $[0, 1]$).
    By pairwise monotonicity of $f$, the following holds.
    $$f(x_1,...,x_n) - f_i(x_1,...,x_{i-1},x_{i+1},...,x_n) \leq 2 \cdot \left(  f(x_1,...,x_n) - f(x_1,...,x_{i-1}, k+1, x_{i+1},...,x_n) \right).$$
    Here, $x_i = k+1$ corresponds to item $i$ not being put into any bucket.
    Finally, the following holds for all $i \in [n]$ due to the marginal gains property.
    $$f(x_1,...,x_n) - f(x_1,...,x_{i-1}, k+1, x_{i+1},...,x_n) \leq f(x_1, \cdots, x_{i}, k+1, \cdots, k+1) - f(x_1, \cdots, x_{i-1}, k+1, \cdots, k+1).$$
    This yields the desired condition.
    \begin{align*}
        &\sum\limits_{i=1}^n(f(x_1, \cdots ,x_n)-f_i(x_1, \cdots ,x_{i-1},x_{i+1}, \cdots, x_n))\\
        &\leq \sum\limits_{i=1}^n 2 \cdot \left( (f(x_1,\cdots,x_n)-f(x_1,\cdots,x_{i-1},k+1,x_{i+1},\cdots,x_n) \right)\\
        &\leq \sum\limits_{i=1}^n 2 \cdot \left ( f(x_1, \cdots, x_{i}, k+1, \cdots) - f(x_1, \cdots, x_{i-1}, k+1, \cdots) \right)\\
        &\leq 2 \cdot f(x_1, \cdots, x_n).
    \end{align*}
    We have thus shown that $f$ is 2-self--bounding, and by \cref{lemma:k-submod-bounding}, $\mathrm{Var}[f] \leq 2 \cdot \mathbb{E}[f]$.
\end{proof}

\subsection{Weak junta approximation}

We first prove the following ``weak'' junta approximation for $k$--submodular functions.

\begin{theorem}[Weak junta approximation]
    \label{thm:junta-lemma-prelim}
    Let $f :[k+1]^J\to [0,1]$ be a $k$--submodular function and $\varepsilon \in (0, 1)$ be a constant.
    Then there exists a $k$--submodular function $h : [k+1]^J\to [0,1]$ that depends only on $J'\subseteq J$ dimensions, where 
    \[
        |J'| \leq \frac{32}{\varepsilon^2} \cdot \log \frac{8 |J|}{\varepsilon^2},
    \]
    and the function $h$ satisfies
    \[
        \lVert f - h \rVert_2 := \left( \mathbb{E}_{x \in [k+1]^J} \left[ (f(x)-h(x))^2 \right] \right)^{\frac{1}{2}} \leq \frac{1}{2} \cdot \varepsilon.
    \]
\end{theorem}

We prove \cref{thm:junta-lemma-prelim} by first describing a procedure to construct the set $J'$ and then showing that there exists a $J'$-junta satisfying the desired properties. 

Let $S \subseteq J$ and $P : J \to [k+1]^n$ be some partition function over $J$.
We denote by $S_P$ the point on the hypergrid obtained by placing each $i \in S$ into bucket $P(i)$.
We denote by $S(\delta)_P$ a random set generated by sampling each element in $S$ independently with probability $\delta$.

\begin{algorithm}[H]
\begin{algorithmic}[1]
    \caption{Identifying the junta set $J'$}
    \label{alg:junta-greedy-alg}
    \vspace{1mm}
    \State \textbf{Input:} Oracle access to $k$--submodular function $f:[k+1]^J\to [0,1]$; parameters $\alpha, \delta > 0$
    \vspace{1mm}
    \State $S \gets \emptyset$.
    \While{there exists $i \notin S$ such that for all partitions $P$ of $J$ we have $\Pr\left[\Delta_f(S(\delta)_P, P(i), i) > \alpha \right] > \frac{1}{2}$}
        \State $S \gets S \cup \{i\}$
    \EndWhile
    \vspace{1mm}
    \State {\bf return} $J' = S$.
\end{algorithmic}
\end{algorithm}

First we show that the number of indices picked by \cref{alg:junta-greedy-alg} is small. 
\begin{claim}
    \label{clm:junta-greedy-size-bound}
    The set $J'$ output by \cref{alg:junta-greedy-alg} satisfies $|J'|\leq \frac{2}{\alpha\delta}$.
\end{claim}
\begin{proof}
    Let $P$ be an arbitrary partition of $J$, and $X$ be an arbitrary subset of $J$.
    For any $i \in X$, let $X_{<i}$ denote the set of items in $X$ that were added to $S$ before $i$ in \cref{alg:junta-greedy-alg}.
    Define
    \[
        X^+_P := \{ i \in X \mid \Delta_f (X_{<i,P}, P(i), i) > \alpha \}.
    \]
    
    Then by \emph{diminishing marginal gains}, at the end of the execution of \cref{alg:junta-greedy-alg}, we have
    \begin{align*}
        \mathbb{E} [f(S(\delta)_P^+)]
        &= f(\emptyset) + \sum \limits_{i\in S(\delta)_P^+} \Delta_f(S(\delta)_{<i, P}^+, P(i), i)\\
        &> \mathbb{E} [|S(\delta)^+_P|] \cdot \alpha. &[\text{Assuming } f(\emptyset) = 0 \text{ WLOG}]
    \end{align*}
    
    Next, we show that $\mathbb{E} [|S(\delta)^+_P|] \geq \frac{\delta \cdot |S|}{2}$.
    If that is true, then
    \[
        1 \geq \mathbb{E}[f(S(\delta)_P^+)] > \frac{\delta\cdot |S|\cdot \alpha}{2} \implies |S|\leq \frac{2}{\delta\alpha}.
    \]
    To prove our claim, we observe that by construction, for all $u \in S$,
    \[
        \Pr [u \in S(\delta)_P^+] = \Pr [u \in S(\delta)] \cdot \Pr[\Delta_f(S(\delta)_{<u,P}, P(u), u) > \alpha] > \frac{\delta}{2}.
    \]
    Writing $|S(\delta)_P^+| = \sum \limits_{u\in S} X_u$, where $X_u$ is the indicator for whether $u \in S(\delta)_P^+$, we get
    \[
        \mathbb{E}[|S(\delta)_P^+|]=\sum\limits_{u\in S}\Pr[X_u] \geq \frac{\delta \cdot |S|}{2}.
    \]
\end{proof}  

Next, we show that for all partitions $P$ and any ``unimportant'' coordinate $i \in J\setminus J'$, the probability of the marginal gains on top of $J'$ being large with respect to $P$ when adding $i$ to $P(i)$ is very small if we sample elements from $J'$ with probability $\frac{k}{k+1}$ rather than $\delta$.

\begin{claim}
    \label{clm:down-monotone-half-sample}
    For all $i \in J\setminus J'$, all $j \in [k]$, and all partitions $P : J'\to[k]$, the following holds.
    \[
        \displaystyle \Pr \limits_{J'\left(\frac{k}{k+1}\right)} \left[ \Delta_f \left( J'\left(\frac{k}{k+1}\right)_P, j, i \right) > \alpha \right] \leq \left( \frac{1}{k+1} \right)^{\frac{1}{2 \delta}}.
    \]
\end{claim}

\begin{proof}
    Consider the following family of points
    \[
        \mathcal{F} = \left\{ X \in [k+1]^n \mid X \preceq J'_P \text{ and } \Delta_f(X, j, i)>\alpha\right\}
    \]
    This is a down-monotone set -- if $Y \preceq X \in \mathcal{F}$ then $Y \in \mathcal{F}$.
    This is due to \emph{diminishing marginal gains} -- $\Delta_f(Y, j, i) \geq \Delta(X, j, i)$.
    If we let $\sigma_p := \Pr[J'(p)_P \in \mathcal{F}]$ for $p \in (0,1)$, then by \cref{lemma:boosting}, 
    \[
        \sigma_p = (1-p)^{\phi(p)},
    \]
    for some non-decreasing function $\phi$.
    The termination condition of \cref{alg:junta-greedy-alg} guarantees that
    \begin{align*}
        \sigma_\delta := \Pr[J'(\delta)_P \in \mathcal{F}] = (1 - \delta)^{\phi(\delta)} \leq \frac{1}{2},
    \end{align*}
    which implies $\phi(\delta) \geq \frac{1}{2\delta}$ for $\delta \leq \frac{1}{2}$.
    Since $\phi$ is non-decreasing, we have
    \begin{align*}
        \phi\left(\frac{k}{k+1}\right) &\geq \phi(\delta) \geq \frac{1}{2\delta} \implies \\
        \sigma_{\frac{k}{k+1}} &= \left( 1-\frac{k}{k+1} \right)^{\phi\left( \frac{k}{k+1} \right)} \leq \left( \frac{1}{k+1} \right)^{1/2\delta}.
    \end{align*}
\end{proof}

We now define our junta function.
Let $\bar{J} := J \setminus J'$.
Define $h : [k+1]^n \to [0,1]$ as
\begin{equation}
    \label{eq:junta-def}
    h(X) := \mathop{\mathbb{E}} \limits_{Y \in [k+1]^{\bar{J}}} [ f(X_{J'}, Y) ],
\end{equation}
where $x_{J'}$ is the point in the hypergrid where the coordinates in $J'$ are assigned as they were in $J$ and the coordinates in $\bar{J}$ are non-assigned.

First, we show that $h$ is $k$--submodular.
\begin{claim}
    The function $h$ defined in \cref{eq:junta-def} is $k$--submodular.
\end{claim}
\begin{proof}
    We have to show that $h$ satisfies the two following conditions.
    \begin{enumerate}
        \item \emph{Diminishing marginal gains.}
        Let $X, Z \in [k+1]^n$ with $Z \preceq X$ and $e \notin X$.
        Let $j \in [k]$.
        Then, we have
        \begin{align*}
            \Delta_h(Z, j, e) = h(Z_{e}^i)-h(Z) =
            \begin{cases}
                0, &\text{if } j \notin J'\\
                \mathop{\mathbb{E}}\limits_{Y \in [k+1]^{\bar{J}}} [ f((Z_{e}^i)_{J'}, Y) - f(Z_{J'}, Y) ], &\text{otherwise}
            \end{cases}.
        \end{align*}
        If $j \notin J'$, then $\Delta_h(X, j, e) = 0$ as well, so the diminishing marginal gains condition holds.
        Otherwise, by the $k$--submodularity of $f$, we know that for all $y \in [k+1]^{\bar{J}}$, 
        \begin{align*}
            f((Z_{e}^i)_{J'}, Y) - f(Z_{J'}, Y) \geq f((X_{e}^i)_{J'},Y) - f(X_{J'}, Y)
        \end{align*}
        which allows us to conclude that $\Delta_h(Z, j, e) \geq \Delta_h(X, j, e)$, and thus that $h$ satisfies diminishing marginal gains.
        
        \item \emph{Pairwise Monotonicity.}
        Let $X \in [k+1]^n$, $i \neq j \in [k]$ and $e \in X_{k+1}$.
        If $e \notin J$, then $\frac{1}{2} \left( h(X_{e}^i) + h(X_e^j) \right) = h(X)$.
        Otherwise, by the pairwise monotonicity of $f$,
        \begin{align*}
            \frac{1}{2} \left( h(X_{e}^i) + h(X_e^j) \right) = \mathop{\mathbb{E}} \limits_{Y \in [k+1]^{\bar{J}}} \frac{1}{2} [ f((X_{e}^i)_{J'}, Y) - f((X_{e}^j)_{J'}, Y) ] \geq \mathop{\mathbb{E}} \limits_{Y \in [k+1]^{\bar{J}}} [ f(X_{J'}, Y) ] = h(X).
        \end{align*}
    \end{enumerate}
    We can thus conclude by \cref{thm:k-submod-alt-def} that $h$ is $k$--submodular.
\end{proof}

Our final step is to show that $h$ is a good approximation to $f$ in the $\ell_2$-norm. 

\begin{definition}[Bad points]
    \label{def:bad-points}
    A point $X_{J'}$ in the restricted hypergrid is called \emph{bad} if there exists $i \in \bar{J}$ and $j \in [k]$ such that
    \[
        \Delta_f(X_{J'}, j, i) > \alpha,
    \]
    or there exists $i \in X_{J'}$
    \[
        \Delta_f( X_{J'} \setminus i, X_{J'} (i), i) < -\alpha.
    \]
\end{definition}

The $\ell_2$ distance can be bounded as below.
\begin{align*}
    \lVert h - f \rVert_2^2
    &= \mathop{\mathbb{E}} \limits_{X \in [k+1]^{{J}}} [(f(X)-h(X))^2]\\
    &= \mathop{\mathbb{E}} \limits_{X_{J'} \in [k+1]^{{J'}}} \left[ \mathop{\mathbb{E}} \limits_{Y \in [k+1]^{\bar{J}}} \left[ (f(X_{J'}, Y) - \mathop{\mathbb{E}} \limits_{Y' \in [k+1]^{\bar{J}}} \left[ f(X_{J'},Y') \right] )^2 \right] \right] \\
    &= \mathop{\mathbb{E}} \limits_{X_{J'} \in [k+1]^{{J'}}} \left[ \mathop{\mathrm{Var}} \limits_{Y \in [k+1]^{\bar{J}}} \left[ f(X_{J'},Y) \right] \right]\\
    &\leq \Pr[X_{J'} \text{ is \emph{bad}}] + \mathop{\mathbb{E}} \limits_{X_{J'} \in [k+1]^{{J'}} \text{ good}} \left[ \mathop{\mathrm{Var}} \limits_{Y \in [k+1]^{\bar{J}}} \left[ f(X_{J'}, Y) \right] \right]. &[\text{Because } f(X) \leq 1]
    \label{eq:dist-bound}
\end{align*}

We first show that for appropriate parameters to \cref{alg:junta-greedy-alg}, the probability that a point is bad is small.

\begin{claim}
    \label{clm:bad-points-small-prob}
    If we set
    $
        \displaystyle \delta = \frac{\log(k+1)}{2 \log(8 \cdot k \cdot |J| \cdot \varepsilon^{-2})}
    $
    in \cref{alg:junta-greedy-alg}, then
    \begin{align*}
        \Pr \limits_{X_{J'}\in[k+1]^{J'}} [X_{J'}\text{ is \emph{bad}}] \leq \frac{\varepsilon^2}{8}.
    \end{align*}
\end{claim}

\begin{proof}
    Recall that we consider a uniform distribution over $X_{J'} \in [k+1]^{J'}$.
    This distribution is the same as sampling a partition $P : J' \to [k] $ uniformly at random and returning $J'\left(\frac{k}{k+1}\right)_P$.
    To see this, fix an arbitrary point $X \in [k+1]^{J'}$.
    Let $A := \{ i \in J' : x_i = k+1 \}$ and $B = J' \setminus A$.
    The probability of generating this particular $X$ with the latter sampling procedure is
    $$
        \Pr[X] = \left( \frac{1}{k+1} \right)^{|A|}\cdot\left(\frac{k}{k+1}\right)^{|B|}\cdot \left(\frac{1}{k}\right)^{|B|} = \frac{1}{(k+1)^{|J'|}}
    $$
    So, via a union bound, \cref{clm:down-monotone-half-sample} and by plugging in $\delta = \frac{\log(k+1)}{2 \log(8 \cdot k \cdot |J| \cdot \varepsilon^{-2})}$, we get that
    \begin{align*}
        \Pr[X_{J'}\text{ is \emph{bad}}] &= \mathop{\Pr} \limits_{J' \left( \frac{k}{k+1} \right), P : J' \to [k]} \left[ \bigcup\limits_{j \in [k]} \bigcup\limits_{i \in \bar{J}} \left( \Delta_f \left( J' \left( \frac{k}{k+1} \right)_P, j, i \right) > \alpha \right) \right]\\
        &\leq |\bar{J}| \cdot k \cdot \left( \frac{1}{k+1} \right)^{1/2\delta}\\
        &\leq \frac{\varepsilon^2}{8}.
    \end{align*}
\end{proof}

Finally, we show that for appropriate parameters to \cref{alg:junta-greedy-alg}, the variance of $f(X_{J'},\cdot)$ when $X_{J'}$ is good is small.

\begin{claim}
    \label{clm:good-points-small-var}
    If we set
    $
        \displaystyle \alpha = \frac{\varepsilon^2}{16}
    $
    in \cref{alg:junta-greedy-alg}, then
    \[
        \mathop{\mathbb{E}} \limits_{X_{J'} \in [k+1]^{{J'}} \text{ good}} \textrm{Var} \left[ f(X_{J'}, \cdot) \right] \leq \frac{\epsilon^2}{8}.
    \]
\end{claim}

\begin{proof}
    When $X_{J'}$ is good, we have that the marginal gain $\Delta_f(X_{J'}, j, i)$ of adding any $i \in \bar{J}$ to any bucket $j \in [k]$ is at most $\alpha$.
    Because of \emph{diminishing marginal gains}, this bounds the marginal gains of $f(X_{J'}, \cdot)$ over its entire domain by $\alpha$.
    Also, the marginal gains of $f(X_{J'},\cdot)$ are at least $-\alpha$, because if there existed some $Y \geq X_{J'}$, $i \in \bar{J}$ and $j \in [k]$ with $\Delta_f(Y, j, i) < -\alpha$, then by \emph{pairwise monotonicity} we must have for all $j' \neq j$ that $\Delta_f(Y, j', i) \geq \alpha$ as
    $$
        \Delta_f(Y, j, i) + \Delta_f(Y, j', i) \geq 0.
    $$
    But this contradicts the upper bound of $\alpha$ to the marginal gains.
    Thus, $f(X_{J'},\cdot)$ is $\alpha$--Lipschitz and we can appeal to \cref{lemma:self-bounding-variance-bound} to get
    \begin{align*}
        \text{Var}[f(x_{J'},\cdot)] \leq 2 \alpha \cdot \mathbb{E}[f(x_{J'},\cdot)] \leq \frac{\varepsilon^2}{8}.
    \end{align*}    
\end{proof}

We are now ready to finish the proof of \cref{thm:junta-lemma-prelim}.

\begin{proof}[Proof of \cref{thm:junta-lemma-prelim}]
    We let $J'$ by the output of \cref{alg:junta-greedy-alg} with input $f$ and parameters
    \[
        \delta = \frac{\log(k+1)}{2 \log(8 \cdot k \cdot |J| \cdot \varepsilon^{-2})}, \qquad \alpha = \frac{\varepsilon^2}{16}.
    \]
    Then the function $h$ defined by \cref{eq:junta-def} is the desired junta function.
    Indeed, combining \cref{clm:bad-points-small-prob,clm:good-points-small-var}, we get
    \[
        \lVert h-f \rVert_2^2 \leq \frac{\varepsilon^2}{8} + \frac{\varepsilon^2}{8} = \frac{\varepsilon^2}{4}.
    \]
    To bound the size of $J'$, we use \cref{clm:junta-greedy-size-bound} to get
    \[
        |J'| \leq \frac{2}{\alpha \delta} = 32 \cdot \frac{1}{\log(1+k)} \cdot \frac{1}{\varepsilon^2} \cdot \log \frac{8 \cdot k \cdot |J|}{\varepsilon^2} \leq \frac{32}{\varepsilon^2} \cdot \log \frac{8 |J|}{\varepsilon^2}.
    \]
\end{proof}

Using the weak junta approximation, we are able arrive at our final junta approximation.

\subsection{Proof of \cref{thm:junta-thm}}

Let $f : [k+1]^n \to [0, 1]$ be a $k$--submodular function.
We will prove a bound of $|J| \leq \frac{4000}{\epsilon^2} \log \frac{1}{\epsilon^2}$ for the size of the approximating junta.

Observe that this bound holds trivially for $\epsilon < n^{-1/2}$, because then we are allowed to choose $J = [n]$.
For contradiction, suppose that there is $\epsilon \in (n^{-1/2}, 1/2)$ for which the statement of \cref{thm:junta-thm} does not hold.
Let $\mathcal{E} \subset (n^{-1/2}, 1/2)$ be the set of all $\epsilon$ for which the statement does not hold, and pick an $\epsilon \in \mathcal{E}$ such that $\epsilon < 2 \inf \mathcal{E}$.
Then, the statement still holds for $\epsilon_2 = \epsilon^2 < \frac{1}{2} \epsilon$.

By the statement of \cref{thm:junta-thm} for $\epsilon_2$, there is a subset of variables $J'$ of size $|J'| \leq \frac{4000}{\epsilon^2} \log \frac{1}{\epsilon} \leq \frac{2^{13}}{\epsilon^2}$ and a $k$--submodular function $g$ depending only on $J$, such that $\lVert f - g \rVert_2 \leq \epsilon_2 \leq \frac{1}{2} \epsilon$.

Now we apply \cref{thm:junta-lemma-prelim} to $g$ with parameter $\epsilon$.
Thus, there exists a $k$--submodular function $h$ such that $\lVert g - h \rVert_2 \leq \frac{1}{2} \epsilon$, and $h$ depends only on a subset of variables $J'' \subseteq J'$, $|J''| \leq \frac{32}{\epsilon^2} \log \frac{8|J'|}{\epsilon}$.
We have $|J'| \leq \frac{2^{13}}{\epsilon^2}$, and therefore $|J''| \leq \frac{32}{\epsilon^2} \log \frac{2^{16}}{\epsilon^3} \leq \frac{32}{\epsilon^2} \log \frac{1}{\epsilon^3}$ (using $\epsilon \leq \frac{1}{2}$).
We conclude that $|J''| \leq \frac{32}{\epsilon^2} \log \frac{1}{\epsilon^3} \leq \frac{4000}{\epsilon^2} \log \frac{1}{\epsilon}$ as required in \cref{thm:junta-thm}.
By the triangle inequality, we have $\|f - h\|_2 \leq \|f - g\|_2 + \|g - h\|_2 \leq \frac{1}{2} \epsilon + \frac{1}{2} \epsilon = \epsilon$.

However, this would mean that the statement of \cref{thm:junta-thm} holds for $\epsilon$ as well, which is a contradiction. \qed

\section{Testing by implicit-learning in the hypergrid}
\label{sec:testing-implicit-learning}
Now that we have established \cref{thm:junta-thm}, it remains to turn junta structure on the hypergrid into a tester.
We prove an implicit-learning theorem directly for functions on $[k+1]^n$: the tester identifies a small set of influential coordinates and then checks whether the observed values are consistent with a junta close to the target property.

\begin{theorem}[Reminder of \cref{thm:implicit-learning-testing}]
    Let $p \geq 1$. Suppose $\mathcal{P} \subseteq \{f:[k+1]^n \to [0,1]\}$ is a property of functions such that for any $f \in \mathcal{P}$ there exists a $\kappa$-junta $h :[k+1]^n \to [0,1]$ satisfying $||f-h||_2 \leq \varepsilon$. Then, there exists a tester that, given query access to an unknown function $f$, will make {\color{black}$\varepsilon^{-9}\cdot 2^{O(\kappa\log \kappa)}$} queries to $f$ and with probability at least $2/3$ accept if $f \in \mathcal{P}$ and reject if $f$ is $\varepsilon$-far from $\mathcal{P}$ in the $\ell_2$ norm.
\end{theorem}

Once established, \cref{thm:implicit-learning-testing} can be combined with \cref{thm:junta-thm} to give us an $\ell_2$-tester for $k$-submodularity.
The following fact allows us to extend the statement to any $p \geq 1$.

\begin{lemma}
For any $p \geq 1$, let $Q_p(\mathcal{P},\varepsilon)$ be the query complexity of $\ell_p$-testing property $\mathcal{P}$ with proximity parameter $\varepsilon$.
Then
$$
    Q_p(\mathcal{P},\varepsilon) \leq Q_2(\mathcal{P},\varepsilon^{p/2}).
$$
\end{lemma}

This yields the following corollary.

\begin{corollary}
    For all $p \geq 1$ there exists an $\ell_p$-tester for $k$-submodularity making $(\frac{p}{\varepsilon})^{2}\cdot 2^{(\frac{p}{\varepsilon})^{2}\log(\frac{p}{\varepsilon})}$ queries. 
\end{corollary}

We set out to prove \cref{thm:implicit-learning-testing}.
The tester for properties $\mathcal{P}$ that are close to juntas on the hypergrid is given in \cref{alg:testing-by-implicit-learning}.
Let $\mathcal{F}$ be the set of $\kappa$-junta functions that are close to some function in $\mathcal{P}$.
The tester determines whether the input function $f$ is close to some junta function in $\mathcal{F}$.
It does so by learning a set $J$ of $\kappa$ high influence coordinates and checking by brute force if $f$ could match a junta in $\mathcal{F}$ on $J$.
If so, the tester accepts and otherwise it rejects.

\subsection{Setup} 
The proof relies on the theory of Fourier Analysis for functions on the hypergrid.
We include, as part of \cref{appx:fourier-analysis-prelims}, a self-contained treatment of the important concepts and lemmata required.
Our main inspiration is the classic text by O'Donnell \cite{o2021analysis}.
We refer the reader unfamiliar with these concepts to that part of the text first, as they will be used extensively in the analysis to follow.
In particular, we call the reader's attention to the definitions of influence and juntas, which we repeat here for structural cohesion.

\begin{definition}[Influence]
    We define the \textbf{influence} of a set of coordinates $S \subseteq [n]$ as
    \begin{align}
        I_f(S) := \sum \limits_{a \in \mathbb{N}^n_{\leq k}: \text{supp}(a)\not\subseteq \bar{S}} \hat{f}(a)^2,
    \end{align}
    where $\bar{S} := [n]\setminus S$.
\end{definition}

Influences can be estimated by direct sampling, as detailed in \cref{alg:influence-estimation}.

\begin{lemma}
\label{lemma:infl-estimation-guarantee}
Algorithm \ref{alg:influence-estimation} returns an estimate $\widehat{I}_f^{(m)}(S)$ such that:
$$
\Pr\left[|\widehat{I}_f^{(m)}(S)-I_f(S)|\geq \varepsilon\right] \leq 2e^{-2\varepsilon^2 m}
$$
\end{lemma}

\begin{definition}[Junta]
A function $f:[k+1]^n\to\mathbb{R}$ is a $\kappa$-junta on the set $J$ where $|J| = \kappa$ if for every $x,y \in [k+1]^n$ that satisfy $x_i = y_i$ for all $i \in J$ we have $f(x) = f(y)$.
\end{definition}

Various facts about juntas and influences that are referred to in the analysis are included in \cref{appx:fourier-analysis-prelims} with proofs. 

Vital to communicating our algorithm will also be the concept of a \textit{core}, as described by \cite{blais2016testing}:
\begin{definition}[Core of a class of juntas]
Let $\mathcal{F}$ be a collection of $\kappa$-\textit{junta} functions $f:[k+1]^n \to \mathbb{R}$. We define $\mathcal{F}_{\text{core}}$ as the set of functions $h:[k+1]^{\kappa}\to\mathbb{R}$ for which there exists $f \in \mathcal{F}$ and some projection $\phi:[k+1]^n\to[k+1]^{\kappa}$ such that for all $x \in [k+1]^n$ we have that $f(x) = h(\phi(x))$.
\end{definition}

We also consider rounding a function's values up to the nearest multiple of a small constant. If $\mathcal{F}$ is a class of functions then $\mathcal{F}^{(\varepsilon)}$ contains for each $f \in \mathcal{F}$ a function that approximates $f(x)$ with the nearest multiple of $\varepsilon$. 

\begin{algorithm}
\begin{algorithmic}[1]
    \caption{Testing ``near'' junta properties via implicit learning}
    \label{alg:testing-by-implicit-learning}
    \vspace{1mm}
    \State \textbf{Input: }Query access to unknown function $f:[k+1]^n\to[0,1]$. Error $\varepsilon \in (0,1/2)$, Junta parameter $\kappa \geq 1$.
    \State \textbf{Parameters}: {\color{black} Let $L = 20\kappa^4$, $m = \frac{\kappa^2}{\varepsilon^4}$ and $q=\frac{2^{O(\kappa)}}{\varepsilon^5}$}
    \vspace{2mm}
    
    \State Draw points $x^{(1)},...,x^{(q)}$ uniformly at random from $[k+1]^n$.
    \State For each \textbf{label} $\ell \in [k+1]^q$ define its \textbf{group} $G_\ell := \{i \in [n]:x^{(1)}_i=\ell_1,...,x^{(q)}_i=\ell_q\}$.
    \State Let $P_1,...,P_{L}$ be a random equi-partition of the labels $[k+1]^q$. 
    \vspace{2mm}
    \State {\color{black} \underline{Initial Pruning Phase}}
    \vspace{1mm}
    \State Find $\kappa$ parts $P_{i_1},...,P_{i_\kappa}$ for which the approximate influence of the remaining coordinates is minimized:
    $$
    (i_1,...i_\kappa) \gets \mathop{\arg\min}\limits_{J\subset [L]:|J|=\kappa}\widehat{I}^{(m)}_f\left([n]\setminus \bigcup\limits_{j \in J}\bigcup\limits_{\ell \in P_j}G_\ell\right)
    $$
    \State {\color{black} \underline{Tree Pruning Phase}}: Let $(P^{(0)}_{i_1},...,P^{(0)}_{i_\kappa}) \gets (P_{i_1},...,P_{i_\kappa})$
    \vspace{1mm}
    \State For $j \in [\kappa]$, $r \in [R]$, and $s \in \{0,1\}$, define $C_j^{(r)} := \bigcup\limits_{\ell \in P_{i_j}^{(r)}} G_\ell$ and $C_{j,s}^{(r)} := \bigcup\limits_{\ell \in P_{i_j, s}^{(r)}} G_\ell$
    \For{$r =0$ to $R$}
    \State Randomly equipartition each $P^{(r)}_{i_j}$ into two parts $P^{(r)}_{i_j,0}$ and $P^{(r)}_{i_j, 1}$.
    \State Find the optimal way to remove half of all coordinates: Let $s^* \in \{0,1\}^{\kappa}$ be defined as:
    $$
    s^* = \mathop{\arg\min}\limits_{s \in \{0,1\}^{\kappa}} \widehat{I}_f^{(m)}\left([n]\setminus\bigcup\limits_{j=1}^{\kappa} C_{j,s_j}^{(r)}\right)
    $$
    \State Let $P_{i_j}^{(\ell+1)} \gets P_{i_j, s_j^*}^{\ell}$ for all $j \in [\kappa]$
    \EndFor
    \State {\color{black} \underline{Testing Phase}}
    \vspace{1mm}
    \State Let $B \gets \bigcup\limits_{j=1}^{\kappa}C_{j}^{(R)}$ be the final set of coordinates we collected. 
    \If{$\widehat{I}^{(m)}_f\left([n]\setminus B\right) > 5\varepsilon^2$}
    \State \textbf{Output} \textsc{Reject}
    \EndIf

    \State Pick any subset $J = (j_1,...,j_{\kappa})$ of $\kappa$ coordinates arbitrarily from $B$. 
    \State Let $\phi:[k+1]^n\to[k+1]^{\kappa}$ be the projection $\phi(x)_\ell = (x_{J})_{j_\ell}$ for $\ell \in [\kappa]$.
    \State Sample a fresh batch of points $y^{(1)},...,y^{(m)}$ uniformly at random from $[k+1]^n$
    \For{$h \in \mathcal{F}_{\text{core}}^{(\varepsilon)}$}
    \If{$\frac{1}{m}\sum\limits_{i=1}^m(f(y^{(i)})-h(\phi(y^{(i)}))^2 \leq \varepsilon$}
    \State \textbf{Output} \textsc{Accept}
    \EndIf
    \EndFor
    \State \textbf{Output} \textsc{Reject}
\end{algorithmic}
\end{algorithm}

\subsection{Analysis}
Our analysis is structured according to the analysis of \cite{blais2016testing} for the hypercube case. However, we simplify and augment their arguments in certain ways.
\begin{itemize}
    \item We use a fresh sample in the ``Testing Phase'' of the algorithm to avoid a complicated de-correlation argument (Lemma 3.2 in \cite{blais2016testing})
    \item We significantly simplify the proof of Lemma 3.1 (our \cref{lemma:high-infl-B})
    \item We obtain more convenient constants by adjusting the algorithm and analysis.
\end{itemize}

We first show that the set of coordinates $B$ we obtain has high influence:
\begin{lemma}[High influence coordinates are recovered]
\label{lemma:high-infl-B}
Suppose $f$ is $\varepsilon$-close to a $\kappa$ junta. Then the set $B$ obtained by \cref{alg:testing-by-implicit-learning} is such that with probability at least $9/10$,
$$
    I_f([n]\setminus B) \leq 5\varepsilon^2.
$$
\end{lemma}

\begin{proof}
Suppose we condition on the event that \cref{alg:influence-estimation} returns an accurate estimate within additive error $\varepsilon^2/\kappa$ for all the times it is invoked.
To guarantee this with probability at least $5/6$, we set $m = \kappa^2 / \varepsilon^4$ in \cref{lemma:infl-estimation-guarantee}.
Let $g$ be the $\kappa$-junta that $f$ is close to.
WLOG, suppose that $g$ is a junta on the set $[\kappa]$ of coordinates.
We will condition on the event that all the coordinates in $[\kappa]$ lie in different parts of the original equi-partition.
The probability that this event does not happen is at most $\kappa^2 / L \leq \frac{1}{20\kappa^2} \leq 1/20$.

For $0 \leq r \leq R$, consider all the variables eliminated prior to iteration $s$ of the tree pruning stage: 
$$
L_r := [n] \setminus \bigcup\limits_{j=1}^{\kappa} C_{j}^{(r)}$$
We have that:
$$
I_f([n]\setminus B) = I_f(L_R)=I_f(L_0) + \sum\limits_{r=1}^R (I_f(L_r)-I_f(L_{r-1}))
$$
We bound this quantity in steps:
\begin{enumerate}
    \item By \cref{lemma:influence-shift} we have that $I_f([n]\setminus [\kappa]) \leq (\sqrt{I_g([n]\setminus [\kappa])} + \varepsilon)^2 = \varepsilon^2$ since $g$ is a junta on $[\kappa]$. Also, \cref{alg:influence-estimation} returns an estimate with additive error $\varepsilon^2/\kappa$, which means that
    $$
    I_f(L_0) \leq 2\varepsilon^2
    $$
    because influence is a monotone function and \cref{alg:testing-by-implicit-learning} exhaustively examines all subsets of the form $[n]\setminus S$ with $|S| = \kappa\cdot |P_1|$.
    \item We split the iterations $[R]$ into two sets. Let 
    $$
    \mathcal{E}:=\{r \leq R:(L_r\setminus L_{r-1})\cap[\kappa] \neq \emptyset\} 
    $$
    be the set of rounds where a coordinate in $[\kappa]$ is eliminated. We have:
    $$
    \sum\limits_{r=1}^R (I_f(L_r)-I_f(L_{r-1})) = \sum\limits_{r \in \mathcal{E}} (I_f(L_r)-I_f(L_{r-1}))+\sum\limits_{r \notin \mathcal{E}} (I_f(L_r)-I_f(L_{r-1}))
    $$
    \begin{itemize}
        \item For the second term:
        \begin{align*}
        \sum\limits_{r \notin \mathcal{E}} (I_f(L_r)-I_f(L_{r-1})) = \sum\limits_{r\notin \mathcal{E}}\sum\limits_{\substack{a\in\mathbb{N}_{\leq k}^n \\ \text{supp}(a)\not\subseteq \overline{L_{r}}, \text{ supp}(a)\subseteq \overline{L_{r-1}}}} \hat{f}(a)^2\leq \sum\limits_{\substack{a\in\mathbb{N}_{\leq k}^n \\ \text{supp}(a)\not\subseteq [\kappa]}} \hat{f}(a)^2= I_f([n]\setminus [\kappa])\leq \varepsilon^2
        \end{align*}
        \item For the first term, let $S_r := L_r \setminus L_{r-1}$. By the sub-additivity property of influence (\cref{lemma:infl-sub-additive}), we have that:
        \begin{align*}
            \sum\limits_{r \in \mathcal{E}} (I_f(L_r)-I_f(L_{r-1})) \leq \sum\limits_{r \in \mathcal{E}} I_f(S_r) = \sum\limits_{r \in \mathcal{E}}I_f\left(\bigcup\limits_{j=1}^{\kappa}C^{(r)}_{j,1-s^{*}_j}\right)
        \end{align*}
        Our algorithm determines $s^*$ so that the complement influence is minimized.
        Because $r \in \mathcal{E}$, we know that $S_r \cap [\kappa] \neq \emptyset$.
        In fact, we have assumed that $|[\kappa]\cap S_r| = 1$.
        Also, if we let $S_r' := \bigcup\limits_{j=1}^{\kappa}C^{(r)}_{j,s^{*}_j}$ be the set of coordinates we keep, we know in particular that $[\kappa] \cap S_r' = \emptyset$.
        Due to the selection of $S_r$ and the accuracy guarantee of \cref{alg:influence-estimation}, we have:
        $$
        I_f(S_r) \leq \frac{\varepsilon^2}{\kappa} + I_f(S_r')
        $$
        And since $S_r' \subseteq [n]\setminus [k]$, we have we have:
        $$
        \sum\limits_{r \in \mathcal{E}} (I_f(L_r)-I_f(L_{r-1})) \leq \kappa \cdot \frac{\varepsilon^2}{\kappa} + \sum\limits_{r \in \mathcal{E}} I_f(S_r') \leq  \varepsilon^2 + I_f([n]\setminus [\kappa]) \leq 2\varepsilon^2
        $$
    \end{itemize}
\end{enumerate}
Putting it all together, we get that
$$
    I_f([n]\setminus B) \leq 5\varepsilon^2.
$$
\end{proof}

Another useful lemma that we will need says that our calculation in Line 22 approximates the distance of $g$ to $h = h_{\text{core}}\circ \phi$. The proof uses the ideas of \cite{blais2016testing}, adapted for the hypergrid setting. We simplify the proof quite a bit though by using a fresh batch of samples $y^{(1)},...,y^{(m)}$ to avoid the randomization decoupling argument done in the original proof. 
\begin{lemma}
\label{lemma:dist-estimation}
Suppose $f:[k+1]^n \to [0,1]$ satisfies $\dist_2(f,g) \leq \varepsilon$ for some $\kappa$-junta $g$. Then, for every $h_{\text{core}} \in \mathcal{F}_{\text{core}}^{(\varepsilon)}$, if $\phi$ is the mapping defined in Line 19 of \cref{alg:testing-by-implicit-learning}, we have for the function $h = h_{\text{core}}\circ \phi$ that:
$$
\mathop{\Pr}\limits_{y^{(1)},...,y^{(q)}}\left[\left|\left(\frac{1}{q}\sum\limits_{i=1}^q(f(y^{(i)})-h(y^{(i)}))^2\right)^{1/2}-\dist_2(g,h)\right| > 3\varepsilon\right] \leq 2e^{-16q\varepsilon^4}
$$
\end{lemma}

\begin{proof}
The empirical square distance
$$
\dist_X(f,h)^2 := \frac{1}{q}\sum\limits_{i=1}^q(f(y^{(i)})-g(y^{(i)}))^2
$$
is an unbiased estimator of $\dist_2(f,g)$ when the $y^{(i)}$ are drawn independently and uniformly at random. As a result, Hoeffding's inequality establishes that:
$$
\Pr\left[\dist_X(f,g)>\dist(f,g)+\varepsilon\right] \geq 1-e^{-16m\varepsilon^4}
$$
Now, the triangle inequality says that $\dist_X(f,h) \leq \dist_X(f,g)+\dist_X(g,h) \leq 2\varepsilon + \dist_X(g,h)$, except with probability $e^{-16q\varepsilon^4}$. Another application of Hoeffding's inequality gives us that, except with probability $e^{-16q\varepsilon^4}$, the distances $\dist_X(g,h)$ and $\dist(g,h)$ are $\varepsilon$-close, which establishes the claim.
\end{proof}

The following Lemmata express the guarantees of the tester:
\begin{lemma}
\label{lemma:tester-accepts-junta}
Let $f:[k+1]^n \to [0,1]$ and $\mathcal{F}$ be a property of $\kappa$-junta functions on the hypergrid such that there exists some $g \in \mathcal{F}$ for which $\dist(f,g) \leq \varepsilon/16$. Then, \cref{alg:testing-by-implicit-learning} accepts with probability at least $5/6$.
\end{lemma}
\begin{proof}
We will first proceed by substituting $\varepsilon\gets \varepsilon/16$. By \cref{lemma:high-infl-B}, the tester rejects in Line 17 with probability at most $1/10$. Assuming it did not reject, we show that there exists, with probability at least $1/10$, some $h \in \mathcal{F}_{\text{core}}^{(\varepsilon)}$ for which Line 23 accepts. Let us assume WLOG that $g \in \mathcal{F}$ is a $\kappa$-junta on set $[\kappa]$. For the set of coordinates $J$ collected from $B$ in Line 18, let $J' := [k] \cap J$ and assume WLOG that $J' = [\kappa']$. Consider the following chain of functions, as shown in the diagram below: Let $g_{\text{core}} \in \mathcal{F}_{\text{core}}$ be the core of function $g$ and its $\varepsilon$-approximation $h_{\text{core}}$. Let $h = h_{\text{core}} \circ \phi$, where $\phi$ is the mapping defined in Line 19 of \cref{alg:testing-by-implicit-learning}. Also let $h^* \in \mathcal{F}^{(\varepsilon)}$ be the direct $\varepsilon$-approximation of $g$. 

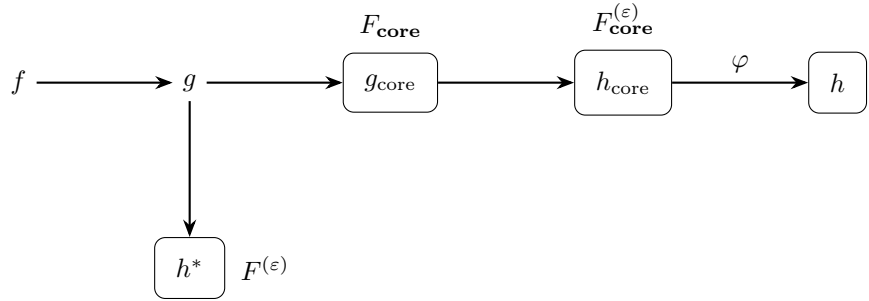
\begin{center}
\begin{tikzpicture}[
  node distance=1.8cm and 1.8cm,
  every node/.style={font=\small},
  box/.style={draw, rectangle, rounded corners, inner sep=8pt},
  arrow/.style={-{Stealth}, thick}
]

\node (f) {\( f \)};
\node (g) [right=of f] {\( g \)};

\node[box, right=of g] (Fcore) {\( g_{\text{core}} \)};
\node[box, right=of Fcore] (Fcoreeps) {\( h_{\text{core}} \)};
\node[box, right=of Fcoreeps] (h) {\( h \)};
\node[box, below=1.8cm of g] (Feps) {\( h^* \)};

\node[above=2pt of Fcore] {\textbf{\( F_{\text{core}} \)}};
\node[above=2pt of Fcoreeps] {\textbf{\( F_{\text{core}}^{(\varepsilon)} \)}};
\node[right=2pt of Feps] {\textbf{\( F^{(\varepsilon)} \)}};

\draw[arrow] (f) -- (g);
\draw[arrow] (g) -- (Fcore);
\draw[arrow] (Fcore) -- (Fcoreeps);
\draw[arrow] (Fcoreeps) -- node[above] {\( \varphi \)} (h);
\draw[arrow] (g) -- (Feps);

\end{tikzpicture}
\end{center}
We have the following bound on $\dist_2(g,h)$:
\begin{align*}
\dist_2(g,h) &\leq \dist_2(g,h^*) + \dist_2(h^*, h) \\
&\leq \varepsilon + \mathop{\mathbb{E}}\limits_{x \sim [k+1]^n }[(h_{\text{core}}(\phi(x))-h^*(x))^2]^{1/2}\\
&
\leq \varepsilon + \mathop{\mathbb{E}}\limits_{x \sim [k+1]^n }[(h_{\text{core}}(x_1,...,x_{\kappa'},x_{i_1},...,x_{i_{\kappa-\kappa'+1}})-h_{\text{core}}(x_1,...,x_{\kappa}))^2]^{1/2}\\
&=\varepsilon + I_{h_{\text{core}}}([\kappa] \setminus [\kappa'])^{1/2}\\
&=\varepsilon + I_{h^*}([n]\setminus [\kappa'])^{1/2}\\
&\leq \varepsilon + I_g([n]\setminus [k'])^{1/2} + \dist_2(g,h^*)\tag{By \cref{lemma:influence-shift}}\\
&\leq \varepsilon + I_g([n]\setminus B)^{1/2} + \dist_2(g,h^*)\\
&\leq \varepsilon + I_g([n]\setminus B)^{1/2} + \dist_2(g,h^*) \tag{$g$ is a junta on $[\kappa]$}\\
&\leq 2\varepsilon + I_f([n]\setminus B)^{1/2} + \dist_2(f,g)\\
&\leq 8\varepsilon \tag{By \cref{lemma:high-infl-B}}
\end{align*}
To finish the proof, we substitute $\varepsilon/16$ in place of $\varepsilon$ and use \cref{lemma:dist-estimation} to argue that our empirical estimator is accurate and thus the tester accepts with probability at least $5/6$.
\end{proof}

\begin{lemma}
\label{lemma:tester-rejects-far-juntas}
Let $f:[k+1]^n \to [0,1]$. If for any $\kappa$-junta $g$ it is true that $\dist_2(f,g) > 6\varepsilon$, \cref{alg:testing-by-implicit-learning} rejects with probability at least $5/6$.
\end{lemma}
\begin{proof}
We will need the following Lemma from \cite{blais2012partially}, which we will state for the hypergrid domain, its proof remaining essentially unchanged:
\begin{lemma}
\label{lemma:equiparition-of-non-juntas}
Let $f:[k+1]^n \to [0,1]$ be $\varepsilon$-far from any $\kappa$-junta. Suppose we partition $[n]$ into $r > \kappa^2$ parts $P_1,...,P_r$. Then with probability at least $5/6$, any union $J$ of $\kappa$ parts from $P$ has $I_f([n]\setminus J) \geq \varepsilon^2/4$.
\end{lemma}
We can invoke this Lemma since we initially partitioned $[n]$ into more than $20\kappa^2$ parts and we're picking $B$ as a subset of a collection of $\kappa$ of those parts. Thus:
$$
I_f([n]\setminus B) \geq 9\varepsilon^2
$$
Assuming accuracy of \cref{alg:influence-estimation} to within $\varepsilon^2/\kappa$, Line 17 rejects with probability at least $5/6$. 
\end{proof}

\begin{lemma}
\label{lemma:tester-rejects-juntas-far-from-P}
Let $f:[k+1]^n \to [0,1]$ be such that for some $\kappa$-junta $g$ it is true that $\dist_2(f,g) \leq 6\varepsilon$. Suppose also that for all $g' \in \mathcal{F}$ we have that $\dist_2(f,g') > \varepsilon$. Then, with probability at least $5/6$, \cref{alg:testing-by-implicit-learning} rejects. 
\end{lemma}

\begin{proof}
We wish to show that the condition of Line 22 is never satisfied. Let $h_{\text{core}} \in \mathcal{F}_{\text{core}}^{(\varepsilon)}$ and $\phi :[k+1]^n \to [k+1]^{\kappa}$ be any projection. Let $h:=h_{\text{core}} \circ \phi \in \mathcal{F}^{(\varepsilon)}$. We aim to establish a lower bound for $\dist_2(g,h)$. Indeed, the triangle inequality gives:
\begin{align*}
    \dist_2(g,h) \geq \dist_2(f,h)-\dist_2(f,g) \geq 7\varepsilon
\end{align*}
Invoking \cref{lemma:dist-estimation} with a union bound over $|\mathcal{F}^{(\varepsilon)}| \leq (1/\varepsilon)^{2^k}$ means that Line 22 will never be satisfied with probability at least $5/6$. 
\end{proof}

Finally, we can put everything together:
\begin{proof}[Proof of \cref{thm:implicit-learning-testing}]
Let $\mathcal{F} = \{g \text{  is a $\kappa$-junta s.t. }\exists f\in \mathcal{P}:||f-g||_2 \leq \varepsilon/16\} $ be the collection of all $\kappa$-juntas close to some function $f \in \mathcal{P}$. This is a property on $\kappa$-junta functions. We run \cref{alg:testing-by-implicit-learning} against $\mathcal{F}$: If $f \in \mathcal{P}$, then, by \cref{lemma:tester-accepts-junta}, \cref{alg:testing-by-implicit-learning} accepts with probability at least $5/6$. If $f$ is $\varepsilon$-far from $\mathcal{P}$ but close to a junta, we reject with probability at least $5/6$ by \cref{lemma:tester-rejects-juntas-far-from-P}, and if $f$ is far from any $\kappa$-junta we also reject with probability at least $5/6$ by \cref{lemma:tester-rejects-far-juntas}. As a result, combining this with the guarantee from \cref{lemma:infl-estimation-guarantee}, if $f$ is $\varepsilon$-far from $\mathcal{P}$, we reject with probability at least $2/3$, as required by the theorem. 
\end{proof}

\section{Towards testing $k$--submodularity in the $\ell_0$-norm}
\label{sec:ksubmod-tester}

In this section, we develop the local-witness theory for the two constraints that characterize $k$--submodularity in Hamming distance.
For diminishing marginal gains, the witnesses are generalized violated squares.
For pairwise monotonicity, the witnesses are violated triangles.
We prove that if either constraint is far from being satisfied, then the corresponding witnesses are dense enough to be found by a non-adaptive one-sided tester.
The key step in both cases is a repair argument over ideals and filters in the hypergrid.

\subsection{Setup}

\begin{definition}[Violating square]
    Let $i, j \in [k]$, $X \in [k+1]^n$, and $e \neq g \in X_{k+1}$.
    We call $\{ X, X_{e}^{i}, X_{g}^{j}, X_{(e,g)}^{(i, j)} \}$ a \emph{square}.
    For a function $f : [k+1]^n \to \mathbb{R}$, this is called an \emph{$f$--violating square} if
    \[
        f(X_{e}^{i}) - f(X) < f(X_{(e,g)}^{(i,j)})- f(X_{g}^{j}),
    \]
    that is, it violates the diminishing marginal gains property.

    \begin{figure}[h!]
    \centering
    \begin{tikzpicture}
        \node (x) at (1,0) {$X$};
        \node (xi) at (0,1) {$X_{g}^{j}$};
        \node (xj) at (2,1) {$X_{e}^{i}$};
        \node (xij) at (1,2) {$X_{(e,g)}^{(i,j)}$};

        \draw[->] (x) -- (xi);
        \draw[->] (x) -- (xj);
        \draw[->] (xi) -- (xij);
        \draw[->] (xj) -- (xij);
    \end{tikzpicture}
    \caption{A square.}
    \end{figure}
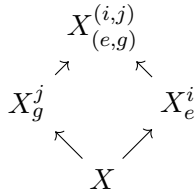
\end{definition}

An inductive argument shows that absence of violating squares is equivalent to the diminishing marginal gains property:

\begin{proposition}
\label{prop:viol-squares}
A function $f:[k+1]^n\to\mathbb{R}$ satisfies the diminishing marginal gains property if and only if it has no violating squares.
\end{proposition}
\begin{proof}
If a function satisfies the diminishing marginal gains property, then it is easy to see that it has no violating squares: take $Y = X_{g}^j$. We show the converse by considering $X \preceq Y$, $e\notin Y$ and $i \in [k]$. Suppose $e_1,...,e_\ell$ are the elements of $Y$ not in $X \sqcap Y$ that have $Y(e) \neq k+1$. Then we have that:
\begin{align*}
    f(X_e^i)-f(X) \geq f(X_{e,e_1}^{i,Y(e_1)})-f(X_{e_1}^{Y(e_1)})
\end{align*}
By induction we can telescopically build the partial partition $Y$ on the right hand side and recover the diminishing marginal gains property. 
\end{proof}

\begin{definition}[Violating triangle]
    Let $X \in [k+1]^n$, $i \neq j \in [k]$, and $e \in X_{k+1}$.
    We call $\{ x, X_{e}^{i}, X_{e}^{j} \}$ a \emph{triangle}.
    For a function $f : [k+1]^n \to \mathbb{R}$, this is called an \emph{$f$--violating triangle} if
    \[
        \frac{f(X_{e}^{i}) + f(X_{e}^{j})}{2} < f(X),
    \]
    that is, it violates the pairwise monotonicity condition.

    \begin{figure}[h!]
    \centering
    \begin{tikzpicture}
        \node (x) at (1,0) {$X$};
        \node (xie) at (0,1) {$X_{e}^{i}$};
        \node (xje) at (2,1) {$X_{e}^{j}$};

        \draw[->] (x) -- (xie);
        \draw[->] (x) -- (xje);
    \end{tikzpicture}
    \caption{A triangle.}
\end{figure}
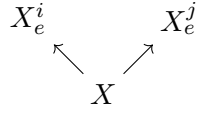
\end{definition}

\subsection{Testing diminishing marginal gains}

\subsubsection{Algorithm}
\begin{algorithm}[H]
\begin{algorithmic}[1]
    \caption{A tester for diminishing marginal gains}
    \label{alg:dim-marg-gain-tester}
    \vspace{1mm}
    \State \textbf{Input: } Parameters $k, n$, a query oracle for function $f:[k+1]^n \to \mathbb{R}$, and an error parameter $\epsilon > 0$.
    \vspace{1mm}
    \For{$n^2 \cdot (1/\varepsilon)^{\Theta(\sqrt{n \log n})}$ iterations}
    \State {\bf sample} the following uniformly at random.
    \Indent
    \State A point $X \in [k+1]^n$ from the hypercube,
    \State two (possibly same) indices $i, j \in [k]$, and
    \State two elements $e \neq g \in X_{k+1}$.
    \EndIndent
    \If{$\{ X, X_{e}^{i}, X_{g}^{j}, X_{(e,g)}^{(i,j)} \}$ is an $f$--violating square}
        \State \textbf{reject}.
    \EndIf
    \EndFor
    \State \textbf{accept}.
\end{algorithmic}
\end{algorithm}

\subsubsection{Analysis}
We prove the following theorem.
\begin{theorem}
\label{thm:dim-marg-gain-test-ana}
    \cref{alg:dim-marg-gain-tester} has the following properties.
    \begin{enumerate}
        \item If $f$ has \emph{diminishing marginal gains}, then \cref{alg:dim-marg-gain-tester} {\bf accepts} with probability $1$.
        \item If $f$ is $\epsilon$-far from \emph{diminishing marginal gains}, then \cref{alg:dim-marg-gain-tester} {\bf rejects} with probability at least $2/3$.
        \item \cref{alg:dim-marg-gain-tester} makes $q = 4 n^2 \cdot (1/\varepsilon)^{\Theta(\sqrt{n \log n})}$ queries.
    \end{enumerate}
\end{theorem}

Items 1 and 3 are easy to verify -- they are direct consequences of \cref{thm:k-submod-alt-def} and the description of \cref{alg:dim-marg-gain-tester}, respectively.
The rest of this section will be dedicated to proving Item 2.

Intuitively, our proof proceeds as follows.
We first show that if a function $f$ has a ``low density'' of violating squares, then it can be made to have the \emph{diminishing marginal gains} property by changing the values of only a ``few'' points.
By contrapositive, this means that a function that is ``far'' from \emph{diminishing marginal gains} must have a ``high density'' of witnesses, which implies the correctness of our tester.
To ``repair'' a function $f$ towards \emph{diminishing marginal gains}, we modify its values over a union of \emph{ideals} and \emph{filters}, defined as follows.

\begin{definition}[Ideal]
    Let $X \in [k+1]^n$.
    The \emph{Ideal} of $X$ is defined as
    \[
        I(X) := \{ Y \mid Y \preceq X \}.
    \]
\end{definition}

\begin{definition}[Filter]
    Let $X \in [k+1]^n$.
    The \emph{Filter} of $X$ is defined as
    \[
        F(X) := \{ Y \mid X \preceq Y \}.
    \]
\end{definition}

We start with the following ``repair'' lemma for violating squares.

\begin{lemma}[Square repair lemma]
    \label{lemma:repair-square}
    Let $f:[k+1]^n\to\mathbb{R}$ be a function in which $\{ X, X_{e}^{i}, X_{g}^{j}, X_{(e,g)}^{(i,j)} \}$ is an $f$--violating square.
    There is a way to change only the values of the ideal $I(X)$ or the filter $F(X_{(e,g)}^{(i,j)})$ to obtain a function $f'$ in which $\{ X, X_{e}^{i}, X_{g}^{j}, X_{(e,g)}^{(i,j)} \}$ is no longer an $f'$--violating square, and no additional $f'$--violating squares have been introduced.  
\end{lemma}
\begin{proof}
    Because $\{ X, X_{e}^{i}, X_{g}^{j}, X_{(e,g)}^{(i,j)} \}$ is an $f$--violating square, we have
    $$
        \Delta_X := f(X) + f(X_{(e,g)}^{(i,j)}) - f(X_{e}^{i}) - f(X_{g}^{j}) > 0.
    $$
    Fix any constant $\Delta \geq \Delta_X$.
    We obtain $f'$ by subtracting $\Delta$ from $f(Y)$ for all $Y \in I(X)$.
    Clearly, this repairs the $f$--violating square itself, but we must also argue that it does not introduce any new $f'$--violating squares. 
    
    Let $\{ Y, Y_{e'}^{i'} ,Y_{g'}^{j'}, Y_{(e', g')}^{(i',j')} \}$ be an $f'$--violating square that is not $f$--violating.
    If $Y \not\preceq X$, then none of the values in the square were modified.
    So we can assume $Y \preceq X$.
    We consider the following cases.
    \begin{enumerate}
        \item Neither $Y_{e'}^{i'} \not\preceq X$ nor $Y_{g'}^{j'} \not\preceq X$.
        Then the points $Y_{e'}^{i'}, Y_{g'}^{j'}, Y_{(e', g')}^{(i',j')}$ are not modified.
        Thus, the value of this square is
        $$
            (f(Y)-\Delta) + f(Y_{(e',g')}^{(i',j')}) - f(Y_{g'}^{j'}) - f(Y_{e'}^{i'}) = f(Y) + f(Y_{(e',g')}^{(i',j')}) - f(Y_{g'}^{j'}) - f(Y_{e'}^{i'}) - \Delta \leq -\Delta.
        $$
        
        \item Only $Y_{e'}^{i'} \not\preceq X$ and $Y_{g'}^{j'} \preceq X$.
        Then the points $Y_{e'}^{i'}, Y_{(e',g')}^{(i',j')}$ are not modified.
        Thus, the value of this square is
        $$
            (f(Y)-\Delta) + f(Y_{(e',g')}^{(i',j')}) - (f(Y_{g'}^{j'})-\Delta) - f(Y_{e'}^{i'}) = f(Y) + f(Y_{(e',g')}^{(i',j')}) - f(Y_{g'}^{j'}) - f(Y_{e'}^{i'}) \leq 0.
        $$
        The case of $Y_{e'}^{i'} \preceq X$ and $Y_{g'}^{j'} \not\preceq X$ is symmetric and identical.
        
        \item Both $Y_{e'}^{i'} \preceq X$ and $Y_{g'}^{j'} \preceq X$.
        Then $Y_{(e', g')}^{(i',j')} \preceq X$, and all the points are modified.
        Thus, the value of this square is
        $$
            (f(Y)-\Delta) + (f(Y_{(e',g')}^{(i',j')})-\Delta) - (f(Y_{g'}^{j'})-\Delta) - (f(Y_{e'}^{i'})-\Delta) = f(Y) + f(Y_{(e',g')}^{(i',j')}) - f(Y_{g'}^{j'}) - f(Y_{e'}^{i'}) \leq 0.
        $$
    \end{enumerate}
    As a result, we have shown that no new violating squares are introduced after changing all the values in $I(X)$, and the original violating square rooted in $X$ is repaired.
    An identical argument suffices to prove that subtracting $\Delta$ from the filter $F(X_{(e,g)}^{(i,j)})$ has the same consequence.
\end{proof}

The proof of this lemma also implies the following corollary.

\begin{corollary}
    \label{cor:removing-viol-squares}
    Let $f:[k+1]^n \to \mathbb{R}$ be a function and $X \in [k+1]^n$ be a domain point.
    For any constant $\Delta > 0$, the function $f'$ obtained by subtracting $\Delta$ from $f(Y)$ for all $Y \in I(X)$ or all $Y \in F(X)$ does not introduce any new $f'$--violating squares.
\end{corollary}

Next, we prove the following density lemma.

\begin{lemma}
    \label{lem:square-density-lemma}
    Let $f : [k+1]^n \to \mathbb{R}$ be a function with $m$ violating squares.
    If $m \leq \varepsilon^{\sqrt{n} \log n} \cdot (k+1)^n$ for $\varepsilon \in (0, e^{-5})$, then $f$ can be made to have \emph{diminishing marginal gains} by changing at most $\varepsilon \cdot (k+1)^n$ values. 
\end{lemma}

\begin{proof}
    The proof is a counting argument on the number of points we will need to change by the square repair procedure above.
    We enumerate all $m$ violating squares and arrange them in order of item-weight (equal to $\sum_{i = 1}^{k} |X_i|$) of the bottom point.
    We split the hypercube evenly based on item-weight.
    It can be verified that this split occurs at weight $\frac{nk}{k+1}$.
    \begin{enumerate}
        \item For violating squares where the bottom vertex has item-weight $\leq \frac{nk}{k+1}$, we repair them by modifying values in the ideal.
        Let $I$ denote the bottom vertex in these violating squares.
        Then the total number of modified points is $\abs{\cup_{X \in I} I(X)}$.

        \item For violating squares where the bottom vertex has item-weight $> \frac{nk}{k+1}$, we repair them by modifying values in the filter.
        Let $F$ denote the top vertex in these violating squares.
        Then the total number of modified points is $\abs{\cup_{X \in F} F(X)}$.
    \end{enumerate}
    We first prove the counting argument for Item 1.
    We estimate this cardinality by combining two different bounds across levels of the hypergrid.
    Let $L_j := \{X \in [k+1]^n:|X|\leq j\}$.
    Then we can write
    \[
        \left|\bigcup\limits_{X \in I} I(X) \right| = \sum\limits_{j=0}^{\frac{nk}{k+1}}\left|\bigcup\limits_{X \in I} (I(X) \cap L_j)\right|.
    \]
    We can bound these terms in two ways.
    First, we have
    \[
        \left| \bigcup\limits_{X \in I} (I(X) \cap L_j) \right| \leq \sum\limits_{X \in I} |I(X) \cap L_j| = \sum\limits_{X\in I} {|X| \choose j} \leq |I| \cdot {\frac{nk}{k+1} \choose j}.
    \]
    Second, there is the more trivial bound of
    \[
        \left| \bigcup\limits_{X \in I} (I(X) \cap L_j) \right| \leq |L_j| = {n \choose j}\cdot k^{j}.
    \]
    This second bound is reminiscent of the binomial distribution $\text{Bin} \left( n, \frac{k}{k+1} \right)$.
    Combining this with a Chernoff bound, we can get an upper bound on the number of points with item-weight far from the mean. 
    Now we just have to use the correct bound depending on $j$.
    \begin{enumerate}
        \item If $j \leq \frac{nk}{k+1}-\alpha\sqrt{\frac{nk}{k+1}}$ for some constant $\alpha$, then
        \begin{align*}
            \sum\limits_{j=0}^{\frac{nk}{k+1}-\alpha\sqrt{\frac{nk}{k+1}}} {n \choose j}\cdot k^{j} &=(k+1)^n \sum\limits_{j=0}^{\frac{nk}{k+1}-\alpha\sqrt{\frac{nk}{k+1}}}{n \choose j}\cdot \left(\frac{k}{k+1}\right)^j \left(\frac{1}{k+1}\right)^{n-j}\\
            &= (k+1)^n \cdot \Pr\left[Z\leq \frac{nk}{k+1}-\alpha \sqrt{\frac{nk}{k+1}}\right] &\left[ Z \sim \text{Bin} \left( n,\frac{k}{k+1} \right) \right]\\
            &= (k+1)^n \cdot \Pr\left[Z\leq (1 - \delta)\frac{nk}{k+1} \right] &\left[ \text{Letting } \delta = \frac{\alpha}{\sqrt{\frac{nk}{k+1}}} \right]\\
            &\leq (k+1)^n \cdot e^{- \delta^2 nk/2(k+1)} &[\text{Chernoff bound}]\\
            &= (k+1)^n\cdot e^{-\alpha^2/2}.
        \end{align*}

        \item If $j > \frac{nk}{k+1} - \alpha\sqrt{\frac{nk}{k+1}}$, then
        \begin{align*}
            \sum\limits_{j = \frac{nk}{k+1} - \alpha\sqrt{\frac{nk}{k+1}}}^{\frac{nk}{k+1}} |I|\cdot {\frac{nk}{k+1}\choose j} &= |I|\cdot  \sum\limits_{j=0}^{\alpha\sqrt{\frac{nk}{k+1}}}{\frac{nk}{k+1}\choose j}\\
            &\leq |I| \cdot \left(\frac{nk}{k+1}\right)^{\alpha\sqrt{\frac{nk}{k+1}}}\\
            &\leq |I| \cdot n^{\alpha \sqrt{n}}\\
            &= |I| \cdot e^{\alpha \sqrt{n} \ln n}.
        \end{align*}
    \end{enumerate}
    
    Finally, setting $\alpha = \frac{1}{2}\ln\varepsilon^{-1}$ yields
    \begin{align*}
        \left|\bigcup\limits_{X \in I} I(X) \right| &= \sum\limits_{j=0}^{\frac{nk}{k+1}}\left|\bigcup\limits_{X \in I} (I(X) \cap L_j)\right| &[\text{Splitting by item-weight}]\\
        &= \sum\limits_{j=0}^{\frac{nk}{k+1} - \alpha\sqrt{\frac{nk}{k+1}}} \left|\bigcup\limits_{X \in I} (I(X) \cap L_j)\right| + \sum\limits_{j = \frac{nk}{k+1} - \alpha\sqrt{\frac{nk}{k+1}}}^{\frac{nk}{k+1}} \left|\bigcup\limits_{X \in I} (I(X) \cap L_j)\right| &[\text{Splitting the sum}]\\
        &\leq \sum\limits_{j=0}^{\frac{nk}{k+1}-\alpha\sqrt{\frac{nk}{k+1}}} {n \choose j}\cdot k^{j} + \sum\limits_{j = \frac{nk}{k+1} - \alpha\sqrt{\frac{nk}{k+1}}}^{\frac{nk}{k+1}} |I| \cdot {\frac{nk}{k+1}\choose j} &[\text{Discussed bounds}]\\
        &\leq (k+1)^n \cdot e^{- \alpha^2/2} + |I| \cdot e^{\alpha \sqrt{n} \ln n} &[\text{From case analysis}]\\
        &\leq (k+1)^n \cdot e^{- \alpha^2/2} + m \cdot e^{\alpha \sqrt{n} \ln n} &[|I| \leq m]\\
        &\leq (k+1)^n \cdot e^{- \alpha^2/2} + (k+1)^n \cdot \varepsilon^{\sqrt{n} \ln n} \cdot e^{\alpha \sqrt{n} \ln n} &[m \leq (k+1)^n \cdot \varepsilon^{\sqrt{n} \ln n}]\\
        &= (k+1)^n \cdot \varepsilon^{\frac{1}{8} \ln \epsilon^{-1}} + (k+1)^n \cdot \varepsilon^{\sqrt{n} \ln n} \cdot \varepsilon^{- \frac{1}{2} \sqrt{n} \ln n} &[\alpha = 1/2 \cdot \ln\varepsilon^{-1}]\\
        &\leq(k+1)^n\cdot \left(\varepsilon^{\frac{1}{8}\ln\varepsilon^{-1}}+\varepsilon^{\frac{1}{2} \sqrt{n} \ln n}\right) &[\text{Simplifying}]\\
        &\leq \frac{\varepsilon}{2} \cdot (k+1)^n. &[\varepsilon \in (0, e^{-5})]
    \end{align*}

    This concludes the counting argument for Item 1.

    The counting argument for Item 2 is almost identical.
    We split the filters by level, and for $j \geq \frac{nk}{k+1} + \alpha\sqrt{\frac{nk}{k+1}}$, the Chernoff-inequality based upper bound still holds.
    For $j < \frac{nk}{k+1} + \alpha\sqrt{\frac{nk}{k+1}}$, we get the following bound.
    \begin{align*}
         \sum\limits_{j = \frac{nk}{k+1}}^{\frac{nk}{k+1} + \alpha\sqrt{\frac{nk}{k+1}}} \left| \bigcup\limits_{X \in F} (F(X) \cap L_j) \right| &\leq \sum\limits_{X \in F} |F(X) \cap L_j|\\
        &= \sum\limits_{j = \frac{nk}{k+1}}^{\frac{nk}{k+1} + \alpha\sqrt{\frac{nk}{k+1}}} \sum\limits_{X\in F} {{n - |X|} \choose {j - |X|}} k^{j - |X|}\\
        &\leq \sum\limits_{j = \frac{nk}{k+1}}^{\frac{nk}{k+1} + \alpha\sqrt{\frac{nk}{k+1}}} \sum\limits_{X\in F} {\frac{n}{k+1} \choose {j - |X|}} k^{j - |X|}\\
        &\leq \sum\limits_{j = \frac{nk}{k+1}}^{\frac{nk}{k+1} + \alpha\sqrt{\frac{nk}{k+1}}} \sum\limits_{X\in F} \left( \frac{n}{k+1} \right)^{j - |X|} k^{j - |X|}\\
        &\leq \sum\limits_{j = \frac{nk}{k+1}}^{\frac{nk}{k+1} + \alpha\sqrt{\frac{nk}{k+1}}} \sum\limits_{X\in F} n^{j - |X|}\\ 
        &\leq \sum\limits_{j = \frac{nk}{k+1}}^{\frac{nk}{k+1} + \alpha\sqrt{\frac{nk}{k+1}}} |F| \cdot n^{j - \frac{nk}{k+1}}\\
        &\leq |F| \cdot n^{\alpha\sqrt{\frac{nk}{k+1}}}\\
        &\leq |F| \cdot n^{\alpha \sqrt{n}}\\
        &= |F| \cdot e^{\alpha \sqrt{n} \ln n}.
    \end{align*}
    The rest of the counting argument is exactly the same and thus omitted.

    In conclusion, we need to modify $\leq \varepsilon \cdot (k+1)^n$ values to make $f$ have \emph{diminishing marginal gains}.
\end{proof}

We are now ready to finish the analysis of our tester.

\begin{proof}[Proof of \cref{thm:dim-marg-gain-test-ana}]
    We prove each part separately.
    \begin{enumerate}
        \item By \cref{thm:k-submod-alt-def}, a function with \emph{diminishing marginal gains} does not have any $f$--violating squares.
        As a result, the tester will always {\bf accept}.
    
        \item By the contrapositive of \cref{lem:square-density-lemma}, a function that is $\varepsilon$-far from \emph{diminishing marginal gains} has at least $(k+1)^n \cdot \varepsilon^{\sqrt{n} \ln n}$ violating squares.
        The total number of squares is ${n \choose 2} \cdot k^2 \cdot (k+1)^{n-2} = \Theta(n^2 (k+1)^{n})$.
        Thus, by the witness lemma, $n^2 \cdot (1/\varepsilon)^{\Theta(\sqrt{n \log n})}$ queries suffice to catch an $f$--violating square with high probability.
    
        \item This follows trivially from the design of the algorithm.
    \end{enumerate}
    This completes the analysis of \cref{alg:dim-marg-gain-tester}.
\end{proof}

\subsection{Testing pairwise monotonicity}

\subsubsection{Algorithm}
\begin{algorithm}[H]
\begin{algorithmic}[1]
    \caption{A tester for pairwise monotonicity}
    \label{alg:pair-mon-tester}
    \vspace{1mm}
    \State \textbf{Input: } Parameters $k, n$, a query oracle for function $f:[k+1]^n \to \mathbb{R}$, and an error parameter $\epsilon > 0$.
    \vspace{1mm}
    \For{$k n \cdot (1/\varepsilon)^{\Theta(\sqrt{n \log n})}$ iterations}
    \State {\bf sample} the following uniformly at random.
    \Indent
    \State A point $X \in [k+1]^n$ from the hypercube,
    \State two indices $i \neq j \in [k]$, and
    \State an element $e \in X_{k+1}$.
    \EndIndent
    \If{$\{ X, X_{e}^{i}, X_{e}^{j} \}$ is an $f$--violating triangle}
        \State \textbf{reject}.
    \EndIf
    \EndFor
    \State \textbf{accept}.
\end{algorithmic}
\end{algorithm}

\subsubsection{Analysis}

We prove the following theorem.
\begin{theorem}
\label{thm:pair-mon-test-ana}
    \cref{alg:pair-mon-tester} has the following properties.
    \begin{enumerate}
        \item If $f$ has \emph{pairwise monotonicity}, then \cref{alg:pair-mon-tester} {\bf accepts} with probability $1$.
        \item If $f$ is $\epsilon$-far from \emph{pairwise monotonicity}, then \cref{alg:pair-mon-tester} {\bf rejects} with probability at least $2/3$.
        \item \cref{alg:pair-mon-tester} makes $q = 3 k n \cdot (1/\varepsilon)^{\Theta(\sqrt{n \log n})}$ queries.
    \end{enumerate}
\end{theorem}

Items 1 and 3 are easy to verify -- they are direct consequences of \cref{thm:k-submod-alt-def} and the description of \cref{alg:pair-mon-tester}, respectively.
The rest of this section will be dedicated to proving Item 2.

The proof sketch is very similar to that of the previous section.
We first show that if a function $f$ has a ``low density'' of violating triangles, then it can be made to have the \emph{pairwise monotonicity} property by changing the values of only a ``few'' points.
By contrapositive, this means that a function that is ``far'' from \emph{pairwise monotonicity} must have a ``high density'' of witnesses, which implies the correctness of our tester.
To ``repair'' a function $f$ towards \emph{pairwise monotonicity}, we modify its values over a union of \emph{ideals} and \emph{filters}.

We start with the following ``repair'' lemma for violating triangles.

\begin{lemma}[Triangle repair lemma]
    \label{lemma:repair-triangle}
    Let $f:[k+1]^n\to\mathbb{R}$ be a function in which $\{ X, X_{e}^{i}, X_{e}^{j} \}$ is an $f$--violating triangle.
    There is a way to change only the values of the ideal $I(X)$ or the filter $F(X_{e}^{i})$ or $F(X_{e}^{j})$ to obtain a function $f'$ in which $\{ X, X_{e}^{i}, X_{e}^{j} \}$ is no longer an $f'$--violating triangle, and no additional $f'$--violating triangles have been introduced.  
\end{lemma}
\begin{proof}
    Because $\{ X, X_{e}^{i}, X_{e}^{j} \}$ is an $f$--violating triangle, we have
    $$
        \Delta_X := 2 \cdot f(X) - f(X_{e}^{i}) - f(X_{e}^{j})  > 0.
    $$
    Fix any constant $\Delta \geq \Delta_X$.
    We obtain $f'$ by subtracting $\Delta$ from $f(Y)$ for all $Y \in I(X)$.
    Clearly, this repairs the $f$--violating triangle itself, but we must also argue that it does not introduce any new $f'$--violating triangles. 
    
    Let $\{ Y, Y_{e'}^{i'} ,Y_{e'}^{j'} \}$ be an $f'$--violating triangle that is not $f$--violating.
    If $Y \not\preceq X$, then none of the values in the triangle were modified.
    So we can assume $Y \preceq X$.
    We consider the following cases.
    \begin{enumerate}
        \item Neither $Y_{e'}^{i'} \not\preceq X$ nor $Y_{e'}^{j'} \not\preceq X$.
        Then the points $Y_{e'}^{i'}, Y_{e'}^{j'}$ are not modified.
        Thus, the value of this triangle is
        $$
            2 \cdot (f(Y)-\Delta) - f(Y_{e'}^{i'}) - f(Y_{e'}^{j'}) = f(Y) - f(Y_{e'}^{i'}) - f(Y_{e'}^{j'}) - 2 \Delta \leq -2 \Delta.
        $$
        
        \item Only $Y_{e'}^{i'} \not\preceq X$ and $Y_{e'}^{j'} \preceq X$.
        Then the point $Y_{e'}^{i'}$ is not modified.
        Thus, the value of this triangle is
        $$
            2 \cdot (f(Y)-\Delta) - (f(Y_{e'}^{j'})-\Delta) - f(Y_{e'}^{i'}) = f(Y) - f(Y_{e'}^{j'}) - f(Y_{e'}^{i'}) - \Delta \leq - \Delta.
        $$
        The case of $Y_{e'}^{i'} \preceq X$ and $Y_{e'}^{j'} \not\preceq X$ is symmetric and identical.
    \end{enumerate}
    As a result, we have shown that no new violating triangles are introduced after changing all the values in $I(X)$, and the original violating triangle rooted in $X$ is repaired.
    A similar argument suffices to prove that adding $\Delta$ to the filter $F(X_{e}^{i})$ or $F(X_{e}^{j})$ has the same consequence.
\end{proof}

The proof of this lemma also implies the following corollary.

\begin{corollary}
    \label{cor:removing-viol-triangles}
    Let $f:[k+1]^n \to \mathbb{R}$ be a function and $X \in [k+1]^n$ be a domain point.
    For any constant $\Delta > 0$, the function $f'$ obtained by subtracting $\Delta$ from $f(Y)$ from all $Y \in I(X)$ or adding $\Delta$ to $f(Y)$ for all $Y \in F(X)$ does not introduce any new $f'$--violating triangles.
\end{corollary}

Next, we prove the following density lemma.

\begin{lemma}
    \label{lem:triangle-density-lemma}
    Let $f : [k+1]^n \to \mathbb{R}$ be a function with $m$ violating triangles.
    If $m \leq \varepsilon^{\sqrt{n} \log n} \cdot (k+1)^n$ for $\varepsilon \in (0, e^{-5})$, then $f$ can be made to have \emph{pairwise monotonicity} by changing at most $\varepsilon \cdot (k+1)^n$ values. 
\end{lemma}

\begin{proof}
    Same as the proof of \cref{lem:square-density-lemma}.
\end{proof}

We are now ready to finish the analysis of our tester.

\begin{proof}[Proof of \cref{thm:pair-mon-test-ana}]
    We prove each part separately.
    \begin{enumerate}
        \item By \cref{thm:k-submod-alt-def}, a function with \emph{pairwise monotonicity} does not have any $f$--violating triangles.
        As a result, the tester will always {\bf accept}.
    
        \item By the contrapositive of \cref{lem:triangle-density-lemma}, a function that is $\varepsilon$-far from \emph{pairwise monotonicity} has at least $(k+1)^n \cdot \varepsilon^{\sqrt{n} \ln n}$ violating triangles.
        The total number of triangles is ${k \choose 2} \cdot n \cdot (k+1)^{n-1} = \Theta(k n (k+1)^{n})$.
        Thus, by the witness lemma, $k n \cdot (1/\varepsilon)^{\Theta(\sqrt{n \log n})}$ queries suffice to catch an $f$--violating triangle with high probability.
    
        \item This follows trivially from the design of the algorithm.
    \end{enumerate}
    This completes the analysis of \cref{alg:pair-mon-tester}.
\end{proof}

\subsection{An obstruction to component-wise local repair}
\label{sec:combining-testers}

The component testers of the previous two sections each rest on a dense witness theorem and a local repair argument.
A natural attempt at a tester for full $k$--submodularity is to apply both component testers in parallel and union-bound the failure probability over the two repair arguments.
We now show that this attempt cannot succeed: there exist configurations on which the repair direction prescribed by a triangle witness is forced to oppose the repair direction prescribed by an overlapping square witness on a shared vertex.

For ideals rooted at item-weight $\leq \frac{nk}{k+1}$, the repairs for both witness types proceed downward and can be aligned.
The obstruction lives in the upward direction.
On filters rooted at item-weight $> \frac{nk}{k+1}$, repairing a violated square requires subtracting $\Delta > 0$ from a vertex, whereas repairing a violated triangle in the same filter requires \emph{adding} $\Delta > 0$ to that vertex.
The configuration in Figure~\ref{fig:repair-deadlock} realizes both demands simultaneously on the vertices $x_2$ and $x_3$.

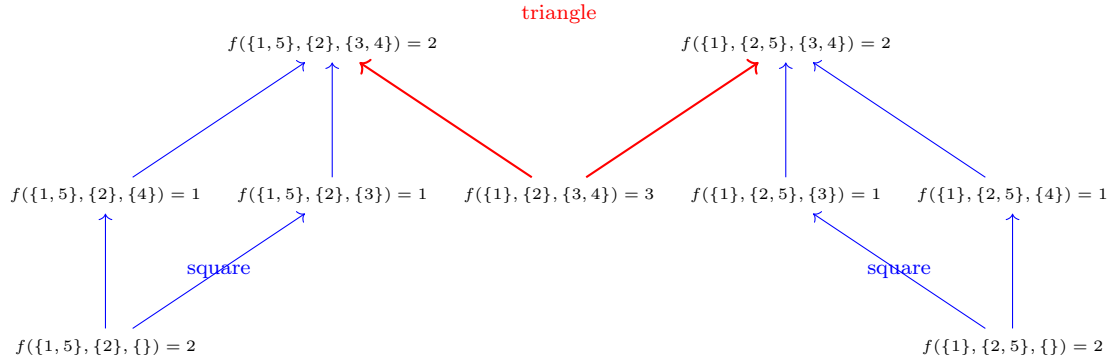
\begin{figure}[h!]
    \centering
    \begin{tikzpicture}
        \tikzstyle{every node}=[font=\tiny]
        \node (x1) at (6,2) {$f(\{ 1 \}, \{ 2 \}, \{ 3, 4 \}) = 3$};
        \node (x2) at (3,4) {$f(\{ 1, 5 \}, \{ 2 \}, \{ 3, 4 \}) = 2$};
        \node (x3) at (9,4) {$f(\{ 1 \}, \{ 2, 5 \}, \{ 3, 4 \}) = 2$};
        \node (x4) at (3,2) {$f(\{ 1, 5 \}, \{ 2 \}, \{ 3 \}) = 1$};
        \node (x5) at (9,2) {$f(\{ 1 \}, \{ 2, 5 \}, \{ 3 \}) = 1$};
        \node (x6) at (0,2) {$f(\{ 1, 5 \}, \{ 2 \}, \{ 4 \}) = 1$};
        \node (x7) at (12,2) {$f(\{ 1 \}, \{ 2, 5 \}, \{ 4 \}) = 1$};
        \node (x8) at (0,0) {$f(\{ 1, 5 \}, \{ 2 \}, \{ \}) = 2$};
        \node (x9) at (12,0) {$f(\{ 1 \}, \{ 2, 5 \}, \{ \}) = 2$};

        \draw[->, red, thick] (x1) -- (x2);
        \draw[->, red, thick] (x1) -- (x3);

        \draw[->, blue] (x4) -- (x2);
        \draw[->, blue] (x6) -- (x2);
        \draw[->, blue] (x8) -- (x4);
        \draw[->, blue] (x8) -- (x6);

        \draw[->, blue] (x5) -- (x3);
        \draw[->, blue] (x7) -- (x3);
        \draw[->, blue] (x9) -- (x5);
        \draw[->, blue] (x9) -- (x7);

        \node[font=\scriptsize, red] at (6, 4.4) {triangle};
        \node[font=\scriptsize, blue] at (1.5, 1) {square};
        \node[font=\scriptsize, blue] at (10.5, 1) {square};
    \end{tikzpicture}
    \caption{A configuration on a high item-weight filter realizing two incompatible local violations. The red edges form a violating \emph{triangle} on item 5 above $x_1$: assigning item 5 to part 1 (giving $x_2$) or to part 2 (giving $x_3$) both yield marginal $-1$, summing to $-2 < 0$. The blue edges form two violating \emph{squares} for diminishing marginal gains on items 3 and 4 in part 3, above $x_8$ (left) and $x_9$ (right). Repairing the triangle requires raising $f(x_2)$ or $f(x_3)$; repairing the squares requires lowering the same values.}
    \label{fig:repair-deadlock}
\end{figure}

A direct check on the values in \Cref{fig:repair-deadlock} confirms both violations. The triangle on item~$5$ above $x_1$ has total marginal $-2$, and the two squares above $x_8$ and $x_9$ each exhibit a strictly increasing marginal as the partial partition grows. Repairing the triangle by a local modification of magnitude $\Delta$ requires \emph{adding} $\Delta$ to $f(x_2)$ or to $f(x_3)$; repairing either square by a local modification of the same magnitude requires \emph{subtracting} $\Delta$ from the corresponding vertex. The two demands are jointly inconsistent on the shared vertices.

The configuration is not isolated to a small region of the hypergrid. Embedding the same value pattern into any high-item-weight filter on a larger domain produces an instance of the same conflict, so an ideal-only repair strategy that avoids upward filters cannot route around it without losing the balance argument that makes the component testers work.

\begin{observation}[Component-wise repair is insufficient]
\label{obs:repair-deadlock}
For every $k \geq 2$ there exist functions $f : [k+1]^n \to \mathbb{R}$ and high item-weight filters on which a violated triangle witness and an overlapping violated square witness prescribe local repairs of opposite sign on a common vertex. In particular, the standard component-wise repair argument---repair triangles and squares independently and combine via union bound---does not yield a tester for $k$-submodularity.
\end{observation}

\Cref{obs:repair-deadlock} does not rule out a Hamming-distance tester for $k$-submodularity; it shows only that obtaining one will require either a richer witness vocabulary that simultaneously certifies both component constraints, or a non-local repair scheme that mediates between square and triangle repairs across the hypergrid. 

\section{Testing in the bounded range setting}
\label{sec:testing-bounded-range}

In this section, we recover efficient Hamming-distance testing for monotone $k$--submodular functions when the range is restricted to $\{ 0, 1, \ldots, r \}$.
The route is representation-theoretic: bounded range and monotonicity force a compact pseudo-DNF structure on the hypergrid, and this structure can be learned.

Our proof has the following outline.
\begin{enumerate}
    \item In \cref{sec:pseudo-dnf-rep-mon-k-submod} we prove that monotone $k$--submodular functions (in fact, monotone functions with diminishing marginal gains) with bounded range can be represented by a pseudo-DNF (which we will define shortly) of bounded width.

    \item In \cref{sec:switching-lem-pseudo-dnf} we prove that any function that can be represented by a pseudo-DNF of bounded width has low-depth decision trees under random restrictions.

    \item In \cref{sec:learn-pseudo-dnf} we relate low-depth decision trees to low-degree Fourier coefficients and prove that such functions have a sparse Fourier approximation.
\end{enumerate}
Combining these results, we get a learning algorithm for monotone $k$--submodularity (using the Kushilevitz-Mansour learning algorithm, which extends to functions defined on the hypergrid).
Formally, we prove \cref{thm:mon-k-submod-bnd-rng-lrn}.

\begin{theorem}[Reminder of \cref{thm:mon-k-submod-bnd-rng-lrn}]
    Let $\eps, \delta \in (0, 1)$ be constants.
    There exists a tester $\cT$ which, given query access to a function $f : [k+1]^n \to \{ 0, 1, \cdots, r \}$ satisfies the following.
    \begin{enumerate}
        \item If $f$ is monotone $k$--submodular, then $\cT$ {\bf accepts} with probability $\geq 1-\delta$.
        \item If $f$ is $\eps$-far from monotone $k$--submodularity, then $\cT$ {\bf rejects} with probability $\geq 1-\delta$.
        \item $\cT$ makes at most $\textsf{poly} (n, 1/\epsilon, 1/\delta, (kr)^{O(rk^2 \log (r/\eps))})$ queries.
    \end{enumerate} 
\end{theorem}

\subsection{Pseudo-DNF representation of monotone $k$--submodular functions}
\label{sec:pseudo-dnf-rep-mon-k-submod}

We first define a pseudo-DNF function on a hypergrid, which is an extension of pseudo-Boolean DNFs defined on the hypercube in \cite{RaskhodnikovaY13}.

\begin{definition}[Pseudo-DNF]
    \label{def:pseudo-dnf}
    A \emph{Pseudo-DNF} of \emph{width $w$} and \emph{size $s$} is an expression of the form
    \[
        f(x_1, \ldots, x_n) = \max_{t = 1}^{s} \left( a_t \cdot \left( \bigwedge\limits_{(i, p) \in A_t} \bm{1}[x_i = p] \right) \right),
    \]
    where $a_t$ are constants, $A_t \subseteq [n] \times [k+1]$, and $|A_t| \leq w$ for $t \in [s]$.
    The class of all pseudo-DNFs of width $w$ with constants $a_t \in \{ 0, \ldots, r \}$ is denoted $\text{DNF}^{w, r}$.
\end{definition}

The difference between pseudo-DNFs and pseudo-Boolean DNFs is that the $x_i$ variables are no longer 1 or 0, so we use indicator functions instead of $x_i$ and $\overline{x_i}$.
Having defined pseudo-DNFs, we are ready to state the main theorem of this subsection.

\begin{theorem}
    \label{thm:k-submod-bnd-dnf}
    Each monotone $k$--submodular function $f : [k+1]^n \to \{ 0, \cdots, r \}$ can be represented by a pseudo-DNF of width $r$ with constants $a_t \in \{ 0, \ldots, r \}$.
\end{theorem}

\begin{proof}
    \cref{alg:k-submod-mon-dnf} below with the initial call $\textsc{Montone-DNF}(f, \emptyset)$ returns the collection of clauses representing the function $f$.
    \begin{algorithm}
    \caption{\textsc{Monotone-DNF}$(f, X)$}
    \label{alg:k-submod-mon-dnf}
    \begin{algorithmic}[1]
    \State {\bf Input:} Oracle access to $f : [k+1]^n \to \{ 0, \dots, r \}$; assignment $X$ of items to buckets
    \State {\bf Output:} Collection $\mathcal{C}$ of clauses of width at most $r$
    \State $\mathcal{C} \gets f(X) \cdot \bigwedge\limits_{(i, p) \in X} \mathbf{1}[x_i = p]$
    \ForAll{$j \in [n] \setminus X$}
      \ForAll{$q \in [k]$}
        \If{$f(X_j^q) > f(X)$}
          \State $\mathcal{C} \gets \mathcal{C} \cup \textsc{Monotone-DNF}(f, X_j^q)$
        \EndIf
      \EndFor
    \EndFor
    \State \Return $\mathcal{C}$
    \end{algorithmic}
    \end{algorithm}
    
    It is easy to see that the width of the clauses is at most $r$, given that the recursion depth is bounded by the range of $f$, and each recursive call adds one item to the assignment $S$.

    Next we prove that the pseudo-DNF $\max_{C_i \in \cC} C_i$ exactly represents $f$.
    By monotonicity of $f$, we have the lower bound $f(X) \geq \max_{C_i \in \cC} C_i (X)$ for any assignment $X \in [k+1]^n$.
    For the upper bound, we consider the smallest subset $Y \subseteq X$ such that $f(X) = f(Y)$.
    Define the set $\cU = \{ Y \setminus \{ i \} \mid i \in Y \}$.
    Then by construction, $f(Z) < f(Y)$ for all $Y \in \cU$.
    Then by the decreasing marginal gains of $f$, the subfunction on the ideal $I(Y)$ is strictly increasing.
    Thus the recursive call $\textsc{Monotone-DNF}(f, Y)$ must have been made at some point, implying $f(X) = \max_{C_i \in \cC} C_i (Y) \leq \max_{C_i \in \cC} C_i (X)$.
    This completes the proof of the lemma.
\end{proof}

\subsection{Switching lemma for pseudo-DNF functions}
\label{sec:switching-lem-pseudo-dnf}

Before stating our switching lemma for functions on the hypergrid, we first define decision trees and random restrictions for such functions, extending the definitions from the hypercube.

\begin{definition}[Decision tree]
    \label{def:dec-tree}
    A decision tree $T$ is a representation of a function $f : [k+1]^n \to \mathbb{R}$.
    It consists of a rooted $(k+1)$-ary tree in which the internal nodes are labelled by coordinates $i \in [n]$, the outgoing edges of each internal node are labelled by a value in $[k+1]$, and the leaves are labelled by real numbers.
    No coordinate $i \in [n]$ appears more than once on a root-to-leaf path.

    Each input $x \in [k+1]^n$ corresponds to a computation path in the tree $T$ from the root to a leaf.
    When the computation path reaches an internal node labelled by a coordinate $i \in [n]$, we say that \emph{$T$ queries $x_i$}.
    The computation path then follows the outgoing edge labelled by by $x_i$.
    The output of $T$ (and hence $f$) is the label of the leaf reached by the computation path.
    We identify a tree with the function it computes.

    The depth $s$ of a decision tree $T$ is the maximum length of any root-to-leaf path in $T$.
    For a function $f$, $\textsf{DT}(f)$ is the minimum depth of a decision tree computing $f$.
\end{definition}

\begin{definition}[Random restriction]
    \label{def:rnd-rest}
    A restriction $\rho$ is a mapping of the input variables of a function to $[k+1] \cup \{ \star \}$.
    The function obtained from $f(x_1, \ldots, x_n)$ by applying a restriction $\rho$ is denoted $f |_{\rho}$.
    The inputs of $f|_{\rho}$ are those $x_i$ for which $\rho(x_i) = \star$ while all other variables are set according to $\rho$.

    A variable $x_i$ is \emph{live} with respect to a restriction $\rho$ if $\rho(x_i) = \star$.
    The set of live variables with respect to $\rho$ is denoted $\textsf{live}(\rho)$.
    A random restriction $\rho$ with parameter $p \in (0, 1)$ is obtained by setting each $x_i$, independently, to a value in $[k+1] \cup \{ \star \}$, so that $\Pr[\rho(x_i) = \star] = p$ and $\Pr[\rho(x_i) = q] = (1-p)/q$ for all $q \in [k+1]$.
\end{definition}

We now state our version of a switching lemma for functions on the hypergrid.

\begin{lemma}[Switching lemma for pseudo-DNFs]
    \label{lem:switching}
    Let $f \in \text{DNF}^{w, r}$ and $\rho$ be a random restriction with parameter $p$.
    Then
    \[
        \Pr[\textsf{DT}(f|_{\rho}) \geq s] < (2 (k+1)^2 p w)^s.
    \]
\end{lemma}

\begin{proof}
    We follow the proof outline of the switching lemma from \cite{RaskhodnikovaY13}, making necessary modifications for the hypergrid, and also removing the dependence on the range of the function.
    Assume an arbitrary total order on the indices, and given the pseudo-DNF $F$ for $f$, arrange the clauses in descending order of their corresponding constant, breaking ties arbitrarily.
    For any restriction $\rho$, the simplified formula $F_{\rho}$ also follows the same ordering of clauses.

    \begin{definition}[Canonical decision tree]
        The canonical labelled decision tree for $F$, denoted $T(F)$ is defined inductively as follows.
        \begin{enumerate}
            \item If $F$ is a constant function, then $T(F)$ consists of a single leaf node labelled by the appropriate constant.

            \item Let $C$ be the first clause of $F$ and $K$ be the set of variables appearing in $C$.
            Then $T(F)$ starts with a complete $(k+1)$-ary tree for $K$, which queries the variables of $K$ in order of the indices.
            Each leaf $v_{\rho}$ of the tree is associated with a restriction $\rho$ of the variables in $K$.
            The unique restriction satisfying $C$ is labelled with the value of the clause $C$.
            The other leaves are replaced by $T(F_{\rho})$, which is defined recursively.
        \end{enumerate}
    \end{definition}

    Let $\mathcal{R}_{n}^{\ell}$ be the set of all restrictions $\rho$ of $n$ variables that has exactly $\ell$ live variables.
    We restate our lemma in terms of $\mathcal{R}_{n}^{\ell}$.
    \begin{lemma}
        Let $F$ be a pseudo-DNF representing a function $f \in DNF^{w, r}$, $s \geq 0$, $p \leq 1/7$, and $\ell = pn$.
        Then
        \[
            \frac{\left| \{ \rho \in \mathcal{R}_n^\ell \mid |T(F_\rho)| \geq s \} \right|}{| \mathcal{R}_n^\ell |} < (2 (k+1)^2 p w)^s.
        \]
    \end{lemma}
    \begin{proof}
        Let $\textsf{stars}(w, s)$ be the set of all sequences $\beta = (\beta_1, \dots, \beta_t)$ such that for each $j \in [t]$, the coordinate $\beta_j \in \{ \star, \_ \}^w \setminus \{ \_ \}^w$ and such that the total number of $\star$'s in all the $\beta_j$'s is $s$.
        Let $S \subseteq \mathcal{R}_{n}^{\ell}$ be the set of restrictions $\rho$ such that $|T(F_\rho)| \geq s$.
        We will define an injective mapping from $S$ to the cartesian product $\mathcal{R}_n^{\ell-s} \times \textsf{stars}(w, s) \times [(k+1)^s] \times [r]$.
        
        Let $F = \max_i C_i$.
        We use $\pi$ to denote both a restriction and a path in the canonical labelled decision tree, which sets variables according to $\pi$.
        Suppose $\rho \in S$ and $\pi$ is the restriction associated with the lexicographically first path in $T(F_\rho)$ of length at least $s$.
        Trim the last variables in $\pi$ along the path $\pi$ so that $|\pi| = s$.
        The image of $\rho$ is defined by following the path $\pi$ in the canonical labelled decision tree for $F_\rho$.
        
        Let $C_{v_1}$ be the first clause of $F$ not falsified by $\rho$.
        Let $K$ be the set of variables in $C_{v_1}|_\rho$ and let $\sigma_1$ be the unique restriction of variables in $K$ that satisfies $C_{v_1}|_\rho$.
        Let $\pi_1$ be the part of $\pi$ that sets the variables in $K$.
        If $\pi_1$ does not set all the variables in $K$, then we shorten $\sigma_1$ to the variables that appear in $\pi_1$.
        Define $\beta_1 \in \{ \star, \_ \}^w$ based on the fixed ordering of the variables in the term $C_{v_1}$ by the letting the $j^{th}$ component of $\beta_1$ be $\star$ if and only if the $j^{th}$ variable of $C_{v_1}$ is set by $\sigma_1$.
        Since $C_{v_1}|_\rho$ is not empty, $\beta_1$ has at least one $\star$.
        From $C_{v_1}$ and $\beta_1$ it is possible to reconstruct $\sigma_1$.

        By definition of $T(F_\rho)$, the restriction $\pi \setminus \pi_1$ labels a path in the canonical labelled decision tree $T(F_{\rho \pi_1})$.
        If $\pi \neq \pi_1$, we repeat the argument above, replacing $\pi$ and $\rho$ with $\pi \setminus \pi_1$ and $\rho \pi_1$ respectively.
        We repeat this process until the round $t$ in which $\pi_1 \pi_2 \dots \pi_t = \pi$.

        Let $\sigma = \sigma_1 \sigma_2 \dots \sigma_t$.
        We define $\delta \in [k+1]^s$ to be a vector that indicates the value of each variable set by $\pi$.
        We define the image of $\rho$ in the injective mapping as a tuple $\langle \rho \sigma_1 \dots \sigma_t, (\beta_1, \dots, \beta_t), \delta \rangle$.

        It remains to show that the defined mapping is injective.
        We will show how to invert it by reconstructing $\rho$ from the image.
        The reconstruction procedure is iterative.
        In one stage of the reconstruction, we recover $\pi_1 \dots \pi_{i-1} \sigma_1 \dots \sigma_{i-1}$ and construct $\rho \pi_1 \dots \pi_{i-1} \sigma_i \dots \sigma_t$.
        Recall that for $i < t$, the restriction $\rho \pi_1 \dots \pi_{i-1} \sigma_i$ satisfies the term $C_{v_i}$ but does not satisfy terms $C_j$ for all $j < v_i$.
        This holds if we extend the restriction by appending $\sigma_{i+1} \dots \sigma_t$.
        Thus, we can recover $v_i$ as the index of the first clause of $F$ that is not falsified by $\rho \pi_1 \dots \pi_{i-1} \sigma_i \dots \sigma_t$.
        Now, based on $C_{v_i}$ and $\beta_i$, we can determine $\sigma_i$.
        Since we know $\sigma_1 \dots \sigma_i$, using the vector $\delta$ we can determine $\pi_i$.
        We can now change $\rho \pi_1 \dots \pi_{i-1} \sigma_i \dots \sigma_t$ to $\rho \pi_1 \dots \pi_{i} \sigma_{i+1} \dots \sigma_t$.
        Finally, given all the values of the $\pi_i$ we reconstruct $\rho$ by removing the variables $\pi_1 \dots \pi_t$ from the restriction.

        The computation in the following claim completes the proof of the switching lemma.
    \end{proof}
\end{proof}

\begin{claim}[Extension of \cite{Bea94}]
    For $p < 1/7$ and $p = \ell / n$, the following holds.
    \[
        \frac{\left| \mathcal{R}_n^{\ell-s} \right| \cdot \left| \textsf{stars} (w, s) \right| \cdot (k+1)^s}{| \mathcal{R}_n^\ell |} < (2 (k+1)^2 p w)^s.
    \]
\end{claim}
\begin{proof}
    We have $| \mathcal{R}_n^\ell | = \binom{n}{\ell} k^{n-\ell}$, so
    \[
        \frac{| \mathcal{R}_n^{\ell-s} |}{| \mathcal{R}_n^\ell |} \leq \frac{(k\ell)^s}{(n-\ell)^s}.
    \]
    We use the following bound on $|\textsf{stars}(w,s)|$
    \begin{proposition}[{\cite[Lemma~2]{Bea94}}]\label{prop-stars-bound}
        $|\textsf{stars}(w,s)| < (w/\ln2)^s$.
    \end{proposition}
    Using \cref{prop-stars-bound}, we get
    \begin{align*}
        \frac{|S|}{| \mathcal{R}_n^\ell |} &\leq \frac{| \mathcal{R}_n^{\ell-s} |}{| \mathcal{R}_n^\ell |} \cdot | \textsf{stars}(w, s) | \cdot (k+1)^s\\
        &\leq \left( \frac{(k+1)^2 \ell w}{(n - \ell) \ln 2} \right)^s\\
        &= \left( \frac{(k+1)^2 p w}{(1 - p) \ln 2} \right)^s.
    \end{align*}
    For $p < 1/7$, the last expression is at most $(2 (k+1)^2 p w)^s$
\end{proof}

\subsection{Learning pseudo-DNF functions}
\label{sec:learn-pseudo-dnf}

In this section, we use \cref{lem:switching} to prove a learning result for pseudo-DNF functions, which implies a tester.
As before, our proof outline follows \cite{RaskhodnikovaY13}, making necessary modifications for the hypergrid.

Let $R_r$ denote the set of multiples of $2/(r - 1)$ in the interval $[-1, 1]$, namely $R_r = \{-1, -1 + 2/(r - 1), \ldots, 1 - 2/(r - 1), 1 \}$.
First, we apply a transformation of the range by mapping $\{0, \dots, r\}$ to $R_r$.
Formally, in this section instead of functions $f : [k+1]^n \to \{ 0, \dots, r \}$ we consider functions $f' : [k+1]^n \to [-1, 1]$ defined as $f'(x_1, \ldots, x_n) = 2/(r - 1) \cdot f(x_1, \ldots, x_n) - 1$.
Since this mapping is bijective, learning algorithms for the class of functions represented by pseudo-Boolean DNF formulas of width $k$ with constants in $R_r$ is sufficient for our result.

Let $\alpha \in \mathbb{N}^n_{\leq k}$ denote a standard Fourier basis vector for $[k+1]^n$ and $\hat{f}(\alpha)$ the Fourier coefficient of $f$ (see definitions in \Cref{appx:fourier-analysis-prelims}).

\begin{definition}
    A function $g$ $\varepsilon$-approximates a function $f$ if $\mathbb{E}[(f - g)^2] \leq \varepsilon$.
    A function is $M$-sparse if it has at most $M$ nonzero Fourier coefficients.
    The Fourier degree of a function, $\deg(f)$, is the size of the vector $\alpha$ with the largest support such that $\hat{f}(\alpha) \neq 0$.
\end{definition}

The following guarantee about approximation of functions in $DNF^{w, r}$ by sparse functions is the key lemma for this section.

\begin{theorem}
    Every function $f \in DNF^{w, r}$ can be $\varepsilon$-approximated by an $M$-sparse function with $M = (wk)^{O(wk^2 \log(2/\varepsilon))}$.
\end{theorem}
\begin{proof}
    We generalize the following key lemmas from \cite{RaskhodnikovaY13}, which relies on multiple applications of the switching lemma.
    \begin{lemma}
        \label{lem:fourier-l2-norm-bound}
        For every function $f \in DNF^{w, r}$,
        \[
            \sum_{\alpha : |\alpha| > 28 w (k+1)^2 \log(2/\varepsilon)} \hat{f}(\alpha)^2 \leq \varepsilon/2.
        \]
    \end{lemma}
    \begin{proof}
        The proof relies on the following result.
        \begin{proposition}[Implicit in \cite{o2021analysis}]
            Let $f : [k+1]^n \to [-1, 1]$ and $f_\rho$ be a random restriction with parameter $p$.
            Then for every $t \in [n]$,
            \[
                \sum_{|\alpha| > t} \hat{f}(\alpha)^2 \leq \Pr_\rho[\deg(f|_\rho) \geq tp/2].
            \]
        \end{proposition}
        Because $\deg(f|_\rho) \leq \mathsf{DT}(f|_\rho)$ and thus $\Pr_\rho[\deg(f|_\rho) \geq tp/2] \leq \Pr_\rho[\mathsf{DT}(f|_\rho) \geq tp/2]$.
        By using \cref{lem:switching} and setting $p = 1/4w(k+1)^2$ and $t = 8 w (k+1)^2 \log (2/\eps)$, we complete the proof.
    \end{proof}
    The next lemma bounds the $\ell_1$-norm of Fourier coefficients corresponding to bounded vectors.
    \begin{lemma}
        \label{lem:fourier-l1-norm-bound}
        For every $f \in \mathsf{DNF}_{k,r}$ and $\tau \in [n]$,
        \[
            \sum_{\alpha: |\alpha| \leq \tau} |\hat{f}(\alpha)| \leq (wk)^{O(\tau)}.
        \]
    \end{lemma}
    \begin{proof}
        Let $L_{1,t}(f) = \sum_{|\alpha| = t} |\hat{f}(\alpha)|$.
        and $L_1 (f) = \sum_{t = 0}^{n} L_{1, t} (f) = \sum_{\alpha} |\hat{f} (\alpha)|$.
        We use the following bound on $L_1(f)$ for decision trees.
        \begin{proposition}[Implicit in \cite{o2021analysis}]
            \label{prop:l1-dec-tree-rel-fourier}
            Consider a function $f : [k+1]^n \to [-1, 1]$, such that $\mathsf{DT}(f) \leq s$.
            Then $L_1(f) \leq (k+1)^s$.
        \end{proposition}
        We show the following generalization of Proposition~4.6 in \cite{RaskhodnikovaY13} for $DNF^{w, r}$.
        \begin{proposition}
            \label{prop:rnd-rest-fourier-l1-norm-bound}
            Let $f \in DNF^{w, r}$ and let $\rho$ be a random restriction of $f$ with parameter $p \leq 1/4w(k+1)^3$.
            Then $\mathbb{E}_\rho [L_1(f|_\rho)] \leq 2$.
        \end{proposition}
        \begin{proof}
            By definition of expectation,
            \[
                \mathbb{E}_\rho [L_1(f|_\rho)] = \sum_{s = 0}^{n} \Pr[\mathsf{DT}(f|_\rho) = s] \cdot \mathbb{E}_\rho [L_1(f|_\rho) \mid \mathsf{DT}(f|_\rho) = s].
            \]
            By \cref{prop:l1-dec-tree-rel-fourier}, for all $\rho$ such that $\mathsf{DT}(f|_\rho) = s$, it holds that $L_1(f) \leq (k+1)^s$.
            By \cref{lem:switching}, $\Pr_\rho[\mathsf{DT}(f|_\rho) \geq s] < (2(k+1)^2pw)^s$.
            Therefore, $\mathbb{E}_\rho [L_1(f|_\rho)] \leq \sum_{s = 0}^n (2(k+1)^2pw)^s \cdot (k+1)^s = \sum_{s = 0}^n (2w(k+1)^3p)^s$.
            For $p \leq 1/4w(k+1)^3$, the proof follows.
        \end{proof}
        Next we use the following generalization of Proposition~4.7 in \cite{RaskhodnikovaY13} to bound $L_{1,t} (f)$.
        The proof is identical and thus omitted.
        \begin{proposition}[Implicit in \cite{RaskhodnikovaY13}]
            \label{prop:rnd-rest-fourier-l1t-norm-bound}
            For $f : [k+1]^n \to [-1, 1]$ and a random restriction $\rho$ with parameter $p$,
            \[
                L_{1, t} (f) \leq \left( \frac{1}{p} \right)^t \mathbb{E}_\rho [L_{1,t}(f|_\rho)].
            \]
        \end{proposition}
        Note that $\sum_{\alpha : |\alpha| \leq \tau} |\hat{f} (\alpha)| = \sum_{t = }^{\tau} L_{1,t}(f)$.
        By setting $p = 1/4w(k+1)^3$ and using \cref{prop:rnd-rest-fourier-l1-norm-bound,prop:rnd-rest-fourier-l1t-norm-bound}, we get
        $L_{1,t}(f) \leq 2(4w(k+1)^3)^t$.
        Thus, $\sum_{\alpha : |\alpha| \leq \tau} |\hat{f} (\alpha)| \leq (wk)^{O(\tau)}$, completing the proof of \cref{lem:fourier-l1-norm-bound}.
    \end{proof}
    
    Let $\tau = 28 w (k+1)^2 \log(2/\varepsilon)$ and $L = \sum_{\alpha : |\alpha| \leq \tau} |\hat{f} (\alpha)|$.
    Let $G = \{\alpha : |\hat{f}(\alpha)| \geq \varepsilon/2L, |\alpha| \leq \tau\}$ and $g(x) = \sum_{\alpha \in G} \hat{f}(S) \phi_\alpha(x)$.
    We will show that $g$ is $M$-sparse and that it $\varepsilon$-approximates $f$.

    By an averaging argument, $|G| \leq 2L^2/\eps$.
    Thus, function $g$ is $M$-sparse, where $M \leq 2L^2/\eps$.
    By \cref{lem:fourier-l1-norm-bound}, $L = (wk)^{O(\tau)} = (wk)^{O(wk^2\log(2/\eps))}$.
    Thus, $M = (wk)^{O(wk^2\log(2/\eps))}$, as claimed in the theorem statement
    
    By definition of $g$ and by Parseval's identity,
    \begin{align*}
        \mathbb{E}[(f - g)^2] &= \sum_{\alpha \not\in G} \hat{f}^2 (\alpha)\\
        &= \sum_{\alpha : |\alpha| > \tau} \hat{f}^2 (\alpha) + \sum_{\alpha : |\alpha| \leq \tau, |\hat{f}^2 (\alpha)| \leq \eps/2L} \hat{f}^2 (\alpha).
    \end{align*}
    By \cref{lem:fourier-l2-norm-bound}, the first summation is at most $\eps/2$.
    For the second summation, we get
    \begin{align*}
        \sum_{\alpha : |\alpha| \leq \tau, |\hat{f}^2 (\alpha)| \leq \eps/2L} \hat{f}^2 (\alpha) &\leq \left( \max_{\alpha : |\hat{f} (\alpha)| \leq \eps/2L} |\hat{f} (\alpha) \right) \left( \sum_{\alpha : |\alpha| \leq \tau} |\hat{f}(\alpha)| \right)\\
        &\leq \frac{\eps}{2L} \cdot L = \eps/2.
    \end{align*}
    This implies that $\mathbb{E}[(f - g)^2] \leq \eps$ and thus $g$ $\varepsilon$-approximates $f$.
\end{proof}

\begin{proof}[Proof of \cref{thm:mon-k-submod-bnd-rng-lrn}]
    From \cref{thm:k-submod-bnd-dnf}, we know that monotone $k$--submodular functions $f \in DNF^{r, r}$
    We will use the following learning algorithm of Kushilevitz and Mansour, which extends to the hypergrid.
    \begin{theorem}[\cite{KushilevitzM93}]
        \label{thm:km93-learn}
        If $f$ can be $\varepsilon$-approximated by an $M$-sparse function, then there is a randomized algorithm running in time $\mathrm{poly}(M, n, 1/\varepsilon, \log(1/\delta))$ that outputs $h$ such that $\mathbb{E}[(f - h)^2] \leq O(\varepsilon)$ with probability at least $1 - \delta$.
    \end{theorem}
    
    Setting the approximation parameter in \cref{thm:km93-learn} to be $\varepsilon' = \varepsilon / Cr^2$ for a large enough constant $C$ and taking $M = (rk)^{O(rk^2\log(2/\eps))}$ we get an algorithm which returns a function $h$ that $(\eps/r^2)$-approximates $f$.
    The running time of such an algorithm is polynomial in $n, (rk)^{O(rk^2\log(2/\eps))}, r^2,$ and $\log(1/\delta)$.
    By \cref{prop:approx-range-change}, if we round the values of $h$ in every point to the nearest multiple of $2/(r-1)$, we will get a function $h'$, such that $\Pr_{x \in [k+1]^n} [h'(x) \neq f(x)] \leq \eps$.
    
    \begin{proposition}
        \label{prop:approx-range-change}
        Let $h : [k+1]^{n} \to [-1,1]$ be an $\varepsilon$-approximation for $f : [k+1]^{n} \to R_r$.
        Let $g$ be the function defined by $g(x) = \arg\min_{y \in R_r} |h(x) - y|$, breaking ties arbitrarily.
        Then $\Pr_{x \in [k+1]^n}[g(x) \neq f(x)] \leq \varepsilon (r - 1)^2$.
    \end{proposition}
    \begin{proof}
    Observe that $|f(x) - h(x)|^2 \geq 1/r^2$ whenever $f(x) \neq g(x)$.
    This implies
    \begin{align*}
        \Pr_{x \in [k+1]^n} [g(x) \neq f(x)] &\leq \Pr_{x \in [k+1]^n} [(r-1)^2 \cdot |f(x) - h(x)|^2 \geq 1]\\
        &\leq \mathbb{E}_{x \in [k+1]^n} [(r-1)^2 \cdot |f(x) - h(x)|^2]\\
        &\leq (r-1)^2 \cdot \mathbb{E}_{x \in [k+1]^n} [|f(x) - h(x)|^2] \leq \eps (r-1)^2,\\
    \end{align*}
    where the last inequality follows from the definition of $\eps$-approximation.
    \end{proof}
    Finally, using the folklore result that a learner with query complexity $q(n,\eps)$ implies a tester with query complexity $q(n,\eps/2)$, we  complete the proof of \cref{thm:mon-k-submod-bnd-rng-lrn}.
\end{proof}

\section*{Acknowledgements} \label{sec:ack}

T.H.\ worked on $k$-submodularity testing as a final project for the 2024 Sublinear Algorithms course taught by Ronitt Rubinfeld. T.H.\ and D.P.\ started working on this project for a course on Sublinear Algorithms taught by Sofya Raskhodnikova \cite{Raskhodnikova25}. We thank Fabian Spaeh for inspiring us to investigate this question through his work on $k$-submodular optimization.


\bibliographystyle{alpha}
\bibliography{bib}

\appendix
\section{Fourier Analysis On the Hypergrid}
\label{appx:fourier-analysis-prelims}

The implicit-learning tester in \cref{sec:testing-implicit-learning} relies on identifying influential coordinates of functions on $[k+1]^n$.
We use Fourier analysis over product domains to define and estimate these influences.
This appendix collects the Fourier-analytic facts used in the tester.
For background on the Boolean case, see O'Donnell \cite{o2021analysis}. 

Let $\Omega := [k+1]$, $\pi = \pi_{\frac{1}{k+1}}$ be the uniform distribution over $\Omega$ and $L^2(\Omega, \pi)$ be the inner product space on the set of functions $\{f:\Omega\to \mathbb{R}\}$ with inner product
$$
    \langle f, g\rangle:=\mathop{\mathbb{E}}\limits_{x \sim \pi_k}[f(x)g(x)].
$$

A \textbf{Fourier basis} $\Phi=\{\phi_0,...,\phi_{k}\}$ is an orthonormal basis of $L^2(\Omega, \pi)$ where $\phi_0 = 1$.
Fourier bases always exist and do not have to be unique, by the Gram–Schmidt process.
Fixing such a basis arbitrarily, define the corresponding \textbf{product basis} as the set of all products between functions in $\Phi$: $\Phi^{\otimes n}:= (\phi_a)_{a \in \mathbb{N}^n_{\leq k}}$, where for $x \in [k+1]^n$ we have
$$
    \phi_a(x) := \prod\limits_{i=1}^n \phi_{a_i}(x_i).
$$

This is an analogue to taking as a basis for all Boolean functions a set indexed by all the points in the hypercube. In the hypergrid we consider all possible assignments of the $n$ items to $k+1$ buckets.
It is easy to check \cite{o2021analysis} that $\Phi^{\otimes n}$ is a Fourier basis for $L^2(\Omega^n, \pi^{\otimes n})$, and so for every $f:[k+1]^n \to \mathbb{R}$ we can write
$$
    f = \sum\limits_{a \in \mathbb{N}^n_{\leq k}} \hat{f}(a)\cdot \phi_a(x),
$$
where $\{\hat{f}(a):a\in \mathbb{N}^n_{\leq k}\}$ are the \textbf{Fourier coefficients} of $f$.
Note that $a = 0$ indexes the constant function $\phi_a = 1$.
We define the support of $a \in \mathbb{N}^n_{\leq k}$ as the set of its non-zero coordinates
$$
    \text{supp}(a) := \{i \in [n]:a_i \neq 0\}.
$$

Expanding the classic results from Boolean function analysis, the following identities hold regardless of the product Fourier Basis we pick for $L^2(\Omega^n, \pi_k^{\otimes n})$.
\begin{enumerate}
    \item \textbf{Parseval's Identity.} $\mathbb{E}[f(x)^2] = \sum\limits_{a \in \mathbb{N}^n_{\leq k}}\hat{f}(a)^2$.
    \item \textbf{Variance Identity.} $\text{Var}[f(x)] = \sum\limits_{a\neq 0}\hat{f}(a)^2$.
    \item \textbf{Plancherel's Identity.} $\langle f,g\rangle = \sum\limits_{a \in \mathbb{N}^n_{\leq k}} \hat{f}(a)\cdot \hat{g}(a)$.
\end{enumerate}

We also define the \textbf{distance} between two functions over $\Omega^n$ in the standard way. For $p \geq 1$, we have:
$$
\dist_p(f,g) = \mathop{\mathbb{E}}\limits_{x \in [k+1]^n}[(f(x)-g(x))^p]^{1/p}
$$

\subsection{Influences in the Hypergrid} Influence for a set of coordinates extends the well known definition for a single coordinate \cite{blais2016testing, o2021analysis}.
\begin{definition}[Influence]
    We define the \textbf{influence} of a set of coordinates $S \subseteq [n]$ as
    \begin{align}
        I_f(S) := \sum \limits_{a \in \mathbb{N}^n_{\leq k}: \text{supp}(a)\not\subseteq \bar{S}} \hat{f}(a)^2,
    \end{align}
    where $\bar{S} := [n]\setminus S$.
\end{definition}

As in the Boolean function case, the influence of a set of coordinates has a number of useful properties.

\begin{lemma}[Influence as expected bias]
\label{lemma:infl-bias}
    For any coordinate subset $S \subseteq [n]$ we have
    \begin{align}
        I_f(S) = \mathop{\mathbb{E}} \limits_{x \sim \pi^{\oplus n}} \left[ \mathop{\text{Var}}\limits_{y \sim \pi^{\oplus S}} \left[ f(x_{\bar{S}},y) \right] \right],
    \end{align}
    where $x_{\bar{S}}$ is the restriction of the coordinates of $x$ to $\bar{S}$.
\end{lemma}

\begin{proof}
    Let $E_S$ be the linear operator that takes the value of $f(x)$ uniformly averaged over the coordinates in $S$, leaving $\bar{S}$ unchanged.
    Define
    $$
        f^{\subseteq \bar{S}}(x):=(E_Sf)(x) = \mathop{\mathbb{E}}\limits_{x'\sim \pi^{\otimes S}}f(x_{\bar{S}},x').
    $$
    As $E_S$ is a linear operator, we have that
    $$
        E_S f = E_S\sum\limits_{a \in \mathbb{N}^n_{\leq k}} \hat{f}(a)\phi_a = \sum\limits_{a \in \mathbb{N}^n_{\leq k}} \hat{f}(a)\cdot E_S\phi_a.
    $$
    And so $\phi_a^{\subseteq \bar{S}} = E_S\phi_a$.
    If $\text{supp}(a) \subseteq \bar{S}$ then $\phi_a$ only depends on $\bar{S}$, which means that $\phi_a^{\subseteq \bar{S}} = \phi_a$.
    Otherwise, we have
    $$
        \phi_a^{\subseteq \bar{S}}(x) = \mathop{\mathbb{E}}\limits_{x'\sim \pi^{\otimes S}} [\phi_{a_{\bar{S}}}(x_{\bar{S}})\cdot \phi_{a_{{S}}}(x')] =\phi_{a_{\bar{S}}}(x_{\bar{S}})\cdot \prod\limits_{i \in S}\mathop{\mathbb{E}}\limits_{x''\sim \pi}\phi_{a_i}(x'').
    $$
    Since $\text{supp}(a)\not\subseteq \bar{S}$ we must have at least some $i \in S$ such that $a_i \neq 0$, so for that $i$ we have by definition that
    $$
        \mathop{\mathbb{E}}\limits_{x''\sim \pi}\phi_{a_i}(x'') = \langle \phi_{a_i},\phi_0\rangle = 0
    $$
    due to the orthonormality of the Fourier basis.
    As a result, we have that
    $$
        E_Sf = \sum\limits_{a \in \mathbb{N}^n_{\leq k}: \text{supp}(a)\subseteq \bar{S}} \hat{f}(a)^2.
    $$
    Now let $L_S f := f-E_S f$.
    We have that
    $$
        L_Sf = \sum\limits_{a \in \mathbb{N}^n_{\leq k}: \text{supp}(a)\not\subseteq \bar{S}} \hat{f}(a)^2.
    $$
    By Parseval's identity we have that
    $$
        I_f(S) = ||L_S f||^2 = \mathop{\mathbb{E}}\limits_{x\sim \pi^{\otimes n}} (f(x)-E_Sf(x))^2.
    $$
    And thus finally we can see that:
    \begin{align*}
        I_f(S) &= \mathop{\mathbb{E}}\limits_{x_{\bar{S}}\sim \pi^{\otimes \bar{S}}} \left[\mathop{\mathbb{E}}\limits_{y\sim \pi^{\otimes S}}(f(x_{\bar{S}},y)-E_Sf(x_{\bar{S}},y))^2\right]\\
        &=\mathop{\mathbb{E}}\limits_{x_{\bar{S}}\sim \pi^{\otimes \bar{S}}} \left[\mathop{\text{Var}}\limits_{y\sim \pi^{\otimes S}}f(x_{\bar{S}},y)\right]\\
        &=\mathop{\mathbb{E}}\limits_{x\sim \pi^{\otimes n}} \left[\mathop{\text{Var}}\limits_{y\sim \pi^{\otimes S}}f(x_{\bar{S}},y)\right],
    \end{align*}
    as claimed.
\end{proof}

The following auxiliary lemma bounds the influence difference of any set $S$ with respect to two functions that are close to each other:
\begin{lemma}[Influence shift under function shift]
\label{lemma:influence-shift}
Let $f,g:[k+1]^n \to [0,1]$ satisfy $||f-g||_2 \leq \varepsilon$. Then for any set $S \subseteq [n]$ we have:
$$
\left|\sqrt{I_f(S)}-\sqrt{I_g(S)}\right| \leq \varepsilon
$$
\end{lemma}
\begin{proof}
Recall that $(E_Sf)(x) = \mathop{\mathbb{E}}\limits_{y \sim \pi^{\otimes S}}f(x_{\bar{S}},y)$. We found that:
$$
I_f(S) = \mathop{\mathbb{E}}\limits_{x\sim \pi^{n}}[(f(x)-(E_Sf)(x))^2] = ||f-(E_Sf)||_2^2
$$
And thus by the triangle inequality we have:
\begin{align*}
\left|\sqrt{I_f(S)}-\sqrt{I_g(S)}\right| &= \left|\,||f-(E_Sf)||_2 - ||g-E_Sg||_2\,\right|\tag{By \cref{lemma:projection-lemma}}\\
&\leq \left|\,||f-(E_Sg)||_2 - ||g-E_Sg||_2\,\right|\\
&\leq ||f-g||_2 \tag{triangle inequality}\\
&\leq \varepsilon
\end{align*}
\end{proof}

\begin{lemma}[Sub-additivity of Influence]
\label{lemma:infl-sub-additive}
Let $S,T \subseteq [n]$. We have that:
$$
I_f(S\cup T) \leq I_f(S) + I_f(T)
$$
\end{lemma}
\begin{proof}
Using the standard definition of influence, we have:
$$
I_f(S\cup T) = \sum\limits_{\text{supp}(a) \cap (S\cap T)\neq \emptyset} \hat{f}(a)^2 \leq \sum\limits_{\text{supp}(a) \cap S\neq \emptyset} \hat{f}(a)^2 + \sum\limits_{\text{supp}(a) \cap T\neq \emptyset} \hat{f}(a)^2 = I_f(S) + I_f(T)\qedhere
$$
\end{proof}

\subsection{Juntas} A junta is a function depending on a small set of coordinates:
\begin{definition}[Junta]
A function $f:[k+1]^n\to\mathbb{R}$ is a $\kappa$-junta on the set $J$ where $|J| = \kappa$ if for every $x,y \in [k+1]^n$ that satisfy $x_i = y_i$ for all $i \in J$ we have $f(x) = f(y)$.
\end{definition}

The function $E_Sf$ we defined in the proof of \cref{lemma:infl-bias} is known with a more general name as the \textit{projection} of $f$ onto coordinate set $\bar{S}$. It is a junta on $\bar{S}$ and in fact is the closest junta to $f$ on $\bar{S}$:
\begin{definition}[Projection to a set of coordinates]
The function:
$$
E_Sf(x) = \mathop{\mathbb{E}}\limits_{y \sim \pi^{\otimes S}}f(x_{\bar{S}},y)
$$
is called the \textbf{projection} of $f$ onto coordinate set $\bar{S}$
\end{definition}

\begin{lemma}
\label{lemma:projection-lemma}
Let $g$ by any junta on $\bar{S}$. Then:
$$
||f-E_Sf||_2 \leq ||f-g||_2
$$
\end{lemma}
\begin{proof}
We claim that $f-E_Sf$ and $E_Sf-g$ are orthogonal in $L^2(\Omega^n, \pi_k^{\otimes n})$. That is because $f-E_Sf$ depends only on variables from $S$ and $E_Sf - g$ depends only on variables from $\bar{S}$, so their inner product is zero. Thus:
$$
||f-g||_2^2 = ||f-E_Sf+E_Sf-g||_2^2 = ||f-E_Sf||_2 + ||E_Sf-g||_2^2 \geq ||f-E_S f||_2^2
$$
\end{proof}

\subsection{Estimating the influence of a set of coordinates}
Given query access to a function $f:[k+1]^n\to [0,1]$ and a set of coordinates $S \subseteq [n]$, we can estimate $I_f(S)$ through a set of examples. This is done via the following algorithm:
\begin{algorithm}[H]
\begin{algorithmic}[1]
\caption{Estimating the influence of a set of variables}
\label{alg:influence-estimation}
\State \textbf{Input:} Query access to function $f:[k+1]^n\to[0,1]$, number of samples $m$.
\State Sample $x_1,...,x_m \sim \pi^{\otimes \bar{S}}$
\State Sample $y_1,...,y_m \sim \pi^{\otimes S}$
\State Sample $y_1',...,y_m' \sim \pi^{\otimes S}$
\State \textbf{Return} $\widehat{I}_f^{(m)}(S):=\frac{1}{m^2}\sum\limits_{i=1}^m\sum\limits_{j=1}^m (f(x_i,y_i)-f(x_i,y_j'))^2$
\end{algorithmic}
\end{algorithm}
\begin{lemma}[\cref{lemma:infl-estimation-guarantee}]
Algorithm \ref{alg:influence-estimation} returns an estimate $\widehat{I}_f^{(m)}(S)$ such that:
$$
\Pr\left[|\widehat{I}_f^{(m)}(S)-I_f(S)|\geq \varepsilon\right] \leq 2e^{-2\varepsilon^2 m}
$$
\end{lemma}
\begin{proof}
Let $\widehat{V}_f(x_i,y_i) := \frac{1}{m}\sum\limits_{j=1}^m (f(x_i,y_i)-f(x_i,y_j'))^2$. We know that $\mathbb{E}_{y_j'}[\widehat{V}_f(x_i,y_i)] = \text{Var}_{y\sim [k+1]^{S}}[f(x_i,y)]$ and we further have that $\widehat{V}_f(x_i,y_i) \in [0,1]$ because the range of $f$ is $[0,1]$. Further, we have that:
\begin{align*}
\mathbb{E}[\widehat{I}_f^{(m)}(S)]=\mathop{\mathbb{E}}\limits_{x_i,y_i,y_j'}\left[\frac{1}{m}\sum\limits_{i=1}^m \widehat{V}_f(x_i,y_i)\right] &= \frac{1}{m}\sum\limits_{i=1}^m \mathop{\mathbb{E}}\limits_{x_i,y_i}\left[\mathop{\mathbb{E}}\limits_{y_j'}[\widehat{V}_f(x_i,y_i)]\right]\\
&= \mathop{\mathbb{E}} \limits_{x \sim \pi^{\oplus n}} \left[ \mathop{\text{Var}}\limits_{y \sim \pi^{\oplus S}} \left[ f(x_{\bar{S}},y) \right] \right]\\
&=I_f(S)
\end{align*}
A straightforward application of Hoeffding's inequality then gives:
$$
\Pr\left[|\widehat{I}_f^{(m)}(S)-I_f(S)|\geq \varepsilon\right]  \leq 2\exp(-2\varepsilon^2 m)
$$
\end{proof}

\end{document}